\documentclass[runningheads]{llncs}

\usepackage[T1]{fontenc}
\usepackage{graphicx}
\usepackage{latexsym}
\usepackage{amssymb}
\usepackage{amsmath}
\usepackage{color}
\usepackage[svgnames]{xcolor}
\usepackage{tikz-cd}
\usepackage{dsfont}
\usepackage{mathrsfs}
\usepackage{makecell}
\usepackage{mathtools}
\usepackage{algorithm}
\usepackage{algorithmic}
\usepackage[T1]{fontenc}
\usepackage{appendix}
\usepackage{cite}
\usepackage{stmaryrd}
\usepackage[hidelinks]{hyperref}
\usepackage{orcidlink}
\usepackage{times}
\usepackage{microtype}
\usepackage{tikz} 
\usepackage{comment}
\usetikzlibrary{shapes.misc,automata,positioning,arrows}
\usepgflibrary{shapes.arrows} 

\title{\mbox{Extending Defeasibility for Propositional Standpoint Logics}}
\author{Nicholas Leisegang\inst{1}\orcidlink{0000-0002-8436-552X}\and
Thomas Meyer\inst{1}\orcidlink{0000-0003-2204-6969}\and
Ivan Varzinczak\inst{2,1}\orcidlink{0000-0002-0025-9632}}
\date{July 2024}

\authorrunning{N. Leisegang et al.}

\institute{University of Cape Town and CAIR, Cape Town, South Africa\\ 
\and
Université Sorbonne Paris Nord\\
Inserm, Sorbonne Université, Limics, 93017 Bobigny, France\\
\email{lsgnic001@myuct.ac.za}, \email{tommie.meyer@uct.ac.za}
\email{ivan.varzinczak@sorbonne-paris-nord.fr}}

\begin{document}

\newcommand{\ie}{\mbox{i.e.}}
\newcommand{\eg}{\mbox{e.g.}}
\newcommand{\cf}{\mbox{cf.}}
\newcommand{\viz}{\mbox{viz.}}
\newcommand{\wrt}{\mbox{w.r.t.}}
\newcommand{\st}{\mbox{s.t.}}
\newcommand{\wolog}{\mbox{w.l.o.g.}}
\newcommand{\etal}{\mbox{et al}}
\newcommand{\aka}{\mbox{a.k.a.}}
\newcommand{\myskip}{\smallskip} 
\newcommand{\df}[1]{\textbf{#1}} 


\newcommand{\Lang}{\mathcal{L}}
\newcommand{\KB}{\mathcal{K}}


\newcommand{\dentails}{\mid\hskip-0.40ex\approx}   
\newcommand{\ndentails}{\not\mid\hskip-0.40ex\approx}

\newcommand{\twiddle}{\leadsto}        
\newcommand{\ntwiddle}{\not\leadsto}   

\newcommand*{\myleftmid}{%
	\mathrel{\vcenter{\offinterlineskip
			\vskip-0.25ex\hbox{$\shortmid$}}}}
\newcommand*{\myrightmid}{%
	\mathrel{\vcenter{\offinterlineskip
			\vskip-0.7ex\hbox{$\shortmid$}}}}
\newcommand*{\twosim}{%
	\mathrel{\vcenter{\offinterlineskip
			\vskip0.05ex\hbox{$\sim$}\vskip0.25ex\hbox{$\sim$}}}}
\newcommand{\bartwosim}{\mathrel{\myleftmid}\hskip-0.03ex\joinrel\twosim}
\newcommand{\dnec}{\mathrel{\bartwosim}\hskip-.05ex\joinrel\myrightmid}
\newcommand{\dposs}{\scalebox{0.8}{\raisebox{-0.2ex}{\rotatebox{57}{\ensuremath{\dnec}}}}\hskip-0.3ex}

\newcommand{\Lflag}{\Lang^{\scalebox{0.5}{${\dnec}$}}}
\newcommand{\Langd}{\widetilde{\Lang}}
\newcommand{\LangdSt}{\Langd_{\mathds{S}}}

\newcommand{\dimp}{\leadsto}
\newcommand{\ndimp}{\not\leadsto}

\newcommand{\SSS}{\mathcal{S}}
\newcommand{\T}{\mathcal{T}}
\newcommand{\B}{\mathcal{B}}

\newcommand{\states}[1]{\llbracket #1 \rrbracket}

\newcommand{\Nick}[1]{\textcolor{orange}{#1}}
\newcommand{\Ivan}[1]{\textcolor{blue}{#1}}
\newcommand{\Tommie}[1]{\textcolor{red}{#1}}


\newcommand{\animal}{\ensuremath{\mathsf{animal}}}
\newcommand{\cheese}{\ensuremath{\mathsf{cheese}}}
\newcommand{\egg}{\ensuremath{\mathsf{egg}}}
\newcommand{\Env}{\ensuremath{\mathsf{Environmentalist}}}
\newcommand{\Pcf}{\ensuremath{\mathsf{Pacifist}}}
\newcommand{\Vga}{\ensuremath{\mathsf{Vegan}}}
\newcommand{\Vgt}{\ensuremath{\mathsf{Vegetarian}}}

\newcommand{\animalF}{\ensuremath{\mathsf{a}}}
\newcommand{\cheeseF}{\ensuremath{\mathsf{c}}}
\newcommand{\eggF}{\ensuremath{\mathsf{e}}}
\newcommand{\EnvF}{\ensuremath{\mathsf{Env}}}
\newcommand{\PcfF}{\ensuremath{\mathsf{Pcf}}}
\newcommand{\VgaF}{\ensuremath{\mathsf{Vga}}}
\newcommand{\VgtF}{\ensuremath{\mathsf{Vgt}}}

\maketitle

\begin{abstract}
In this paper, we introduce a new defeasible version of propositional standpoint logic by integrating Kraus~\etal.'s defeasible conditionals, Britz and Varzinczak's notions of defeasible necessity and distinct possibility, along with Leisegang~\etal.'s approach to defeasibility into the standpoint logics of Gómez Álvarez and Rudolph. The resulting logical framework allows for the expression of defeasibility on the level of implications, standpoint modal operators, and standpoint-sharpening statements. We provide a preferential semantics for this extended language and propose a tableaux calculus, which is shown to be sound and complete with respect to preferential entailment. We also establish the computational complexity of the tableaux procedure to be in \textsc{PSpace}.
\end{abstract}

\section{Introduction}

Standpoint logics are a recently introduced family of agent-centred knowledge representation formalisms~\cite{alvarezrudolph:propositionalstdpt}. Their main feature is to allow the integration of the viewpoints of two or more agents into a single knowledge base, especially when the agents have conflicting takes on a given matter. Standpoint logics are tightly related to various systems of epistemic and doxastic logics since they build on modalities for expressing viewpoints and also assume a Kripke-style possible-worlds semantics. Sentences with standpoint-indexed modal operators such as $\Box_{s}\alpha$ and $\Diamond_{s}\alpha$ read, respectively, ``from the~$s$ standpoint, it is unequivocal that~$\alpha$,'' and ``from the~$s$ standpoint, it is possible that~$\alpha$''. With standpoint-sharpening statements of the form $s\leq t$ (which, in modal-logic terms, is an abbreviation for an axiom schema establishing the interaction between two modalities), one expresses that one standpoint is at least as specific as another, which is a way to say both standpoints agree to some extent.


In spite of allowing for the opinions upheld by agents to be in conflict without causing the knowledge base to be inconsistent, classical standpoint logics do not allow for each agent to handle exceptional cases \emph{within} their respective standpoints. 
This has been partially remedied by Leisegang~\etal.~\cite{LeisegangRudolphMeyer:SacairStandpoints}, who have extended standpoint logics with both defeasible conditionals in the scope of modalities and a non-monotonic form of entailment. The resulting framework, defeasible restricted standpoint logic (DRSL), allows agents to reason about exceptions relative to their own beliefs and for defeasible consequences of a knowledge base to be derived. Nevertheless, DRSL still leaves open the question of a more general approach to defeasibility.

As pointed out by Britz and Varzinczak~\cite{britzvarzin:defeasiblemodalities}, logical languages with modalities make room for exploring defeasibility elsewhere than in conditionals: we can talk of \emph{defeasible necessity} and \emph{distinct possibility}, represented, respectively, by the modal operators~${\dnec}$ and~${\dposs}$.
These enrich modal systems with defeasibility at the object level and meet a variety of applications in reasoning about defeasible knowledge, defeasible action effects, defeasible obligations, and others. It seems only natural that defeasible modalities can be fruitful in providing a formal account of the defeasible standpoints motivated above.

The goal of the present paper is to introduce Propositional Defeasible Standpoint Logic (PDSL), a new defeasible version of standpoint logic enriched with defeasibility aspects on various levels. First, we allow for a defeasible form of implication which is different from the restricted one by Leisegang~\etal.~\cite{LeisegangRudolphMeyer:SacairStandpoints}. Second, drawing on the work of Britz and Varzinczak~\cite{britzvarzin:defeasiblemodalities}, we define defeasible versions of the standpoint modal operators found in classical standpoint logic. Finally, we extend classical standpoint logic further by allowing for the possibility of defeasible standpoint-sharpening statements.

The example below gives an idea of the level of expressivity available in PDSL. 

\begin{example}\label{Example:WorldSaviours}
We consider the standpoints of vegetarians, vegans, pacifists, and environmentalists. From a vegetarian's \emph{usual} standpoint, egg and cheese, although animal-based, are not considered unethical animal products to consume. 
In PDSL, this can be expressed with the sentence ${\dnec}_{\Vgt}((\egg\lor\cheese)\rightarrow\lnot\animal)$, which should not conflict with $\Diamond_{\Vgt}(\egg\land\animal)$, an exception compatible with the vegetarian standpoint which formalises that it is possible (although unusual) to consider an egg an unethical animal product. From the vegan standpoint, though, egg and cheese are unethical animal products. This is formalised with $\Box_{\Vga}((\egg\lor\cheese)\rightarrow\animal)$, and is in line with the intuition that the vegan standpoint is a more specific version of the vegetarian one. This is captured by the sharpening sentence $\Vga\leq\Vgt$. The intuition that \emph{usually}, the vegetarian standpoint is a more specific version of the pacifist one, but allows for exceptions, \eg\ those who do not eat meat only for health reasons, 
can be formalised as the defeasible sharpening $\Vgt\lesssim\Pcf$. Among the consequences of the above, we may expect that there exists a vegan who is a typical representative of the vegan standpoint, who believes that eggs are an unethical animal product and conclude ${\dposs}_{\Vga}(\egg\rightarrow\animal)$. Moreover, while we may expect that typical environmentalists are vegetarians and so $\Env\lesssim\Vgt$ holds, we would not expect that typical environmentalists are necessarily pacifists, and so would expect that $\Env\lesssim\Pcf$ does not hold, even though $\Vgt\lesssim\Pcf$ holds.
\end{example}

The plan of the paper is as follows: Section~\ref{Sec:Preliminaries} recalls the background and notation for the upcoming sections. Following that, and inspired by the work of Britz and Varzinczak~\cite{britzvarzin:defeasiblemodalities} and Leisegang~\etal.~\cite{LeisegangRudolphMeyer:SacairStandpoints}, Section~\ref{section:dmodalities-syntax-semantics} introduces Propositional Defeasible Standpoint Logic (PDSL). In particular, we show that a preferential semantics \emph{à la}~KLM is suitable for interpreting defeasibility in~PDSL and also enables us to define \emph{preferential entailment} \cite{lehmann:conditionalentail} from~PDSL knowledge bases. In Section~\ref{section:satisfiability-and-pref-entailment}, we provide a tableaux-based algorithm for computing preferential entailment for~PDSL, we prove its soundness and completeness, and show that its complexity is in \textsc{PSpace}. Section \ref{Sec:RelatedWork} is a brief discussion on related work. Section~\ref{Sec:Conclusion} closes the paper and considers future work.

\section{Preliminaries}\label{Sec:Preliminaries}

In this section we briefly introduce the basics of classical propositional standpoint logic, as well as defeasible modalities and defeasible reasoning in modal logic,  which form the basis for the logic PDSL introduced in this paper. Standpoint logic was introduced by Gómez Álvarez and Rudolph \cite{alvarezrudolph:propositionalstdpt} for the  propositional case. 
Given a vocabulary $\mathcal{V=(P,S)}$, where $\mathcal{P}$ is a set of propositional atoms and $\mathcal{S}$ a finite set of standpoints including the universal standpoint~$*$, the language $\mathcal{L}_{\mathds{S}}$ over $\mathcal{V}$ is defined by:
    \[\phi::=s\leq t\mid p\mid \neg \phi\mid \phi \wedge \phi\mid \Box_s \phi\]
   where $s,t\in\mathcal{S}$ and $p\in \mathcal{P}$. 
Statements of the form $s\leq t$ are referred to as \textit{standpoint sharpening statements}. The Boolean connectives $\vee$, $\rightarrow$, $\leftrightarrow$ are defined via $\neg$ and $\wedge$ in their usual manner, and for each standpoint $s\in\mathcal{S}$, we define $\Diamond_s$ as $\neg\Box_s\neg$. 

A \emph{standpoint structure} is a triple $M=(\Pi,\sigma,\gamma)$ where 
$\Pi$ is a non-empty set of precisifications;
$\sigma:\mathcal{S}\rightarrow 2^\Pi$ is a function such that $\sigma(*)=\Pi$ and $\sigma(s)\neq \emptyset$ for all $s\in\mathcal{S}$;
$\gamma:\Pi\rightarrow 2^{\mathcal{P}}$ is a function which assigns each precisification a set of atoms. 
Intuitively, the mapping~$\sigma$ allows one to allocate to a standpoint $s$ the set of all ``reasonable ways to make~$s$'s beliefs correct'', and $\gamma$ assigns a set of basic propositions which are `true' in that precisification. 
For a standpoint structure $M$ and a precisification $\pi\in\Pi$, we define the satisfaction relation $\Vdash$ as follows (where $\phi,\phi_1,\phi_2\in\mathcal{L}_{\mathds{S}}$, $s,t\in\mathcal{S}$, and $p\in \mathcal{P}$):
$M,\pi\Vdash p$ iff $p\in\gamma(\pi)$; 
$M,\pi\Vdash \neg \phi$ iff $M,\pi\nVdash \phi$;
$M,\pi\Vdash \phi_1\wedge \phi_2$ iff $M,\pi\Vdash \phi_1$ and $M,\pi\Vdash \phi_2$;
$M,\pi\Vdash \Box_s\phi$ iff $M,\pi'\Vdash \phi$ for all $\pi'\in \sigma(s)$;
$M,\pi\Vdash s\leq t$ iff $\sigma(s)\subseteq \sigma(t)$, and
$M\Vdash\phi$ iff $M,\pi\Vdash \phi$ for all $\pi\in \Pi$.

Defeasible reasoning in modal logic is largely based off of similar methods in the propositional case derived from the notion of preferential consequence relations introduced by Kraus et al. \cite{kraus:nonmonotonic}, and rational consequence relations introduced by Lehmann and Magidor \cite{lehmann:conditionalentail}. Named after the aforementioned authors, this is often called the KLM approach to defeasibility. Preferential consequence relations were considered in the modal case by Britz et al. \cite{britzetal:preferentialreasoningmodal, britzmeyervar:normalmodalpreferential} and extended to include KLM-style defeasibility within modal operators themselves by Britz and Varzinczak\cite{britzvarzin:defeasiblemodalities}. 
In our paper, we build upon the defeasible multi-modal language $\Lflag$ \cite{britzvarzin:defeasiblemodalities}. For a set of propositional atoms $\mathcal{P}$, the language  is $\Lflag$ defined by:
\[
\phi::= p \mid \lnot\phi \mid \phi\land\phi \mid \Box_i\phi \mid {\dnec}_i\phi \mid \phi\twiddle\phi
\]
where $p\in\mathcal{P}$, and $1\leq i \leq n$, for some $n\in\mathds{N}$. The other connectives, $\lor$, $\rightarrow$, and $\leftrightarrow$, are defined as usual. The modality $\Diamond_i$ is defined as $\lnot\Box_i\lnot$, and ${\dposs}_i$ is (analogously) defined as $\lnot{\dnec}_i\lnot$. Intuitively, $\Box_i$ indicates necessity and $\Diamond_i$ possibility (both with respect to~$i$). Regarding the three new operators, ${\dnec}_i$ is intended to indicate ``usual necessity'' (with respect to~$i$), while ${\dposs}_i$ is intended to convey ``distinct'' or ``strong'' possibility (with respect to~$i$), and~$\twiddle$ is a (possibly nested) defeasible conditional. 

A preferential Kripke model is a quadruple $P=(W,R,V,\prec)$, where $W$ is a non-empty set of worlds, $R:=<R_1,\dots,R_n>$, where each $R_i\subseteq W\times W$ is an accessibility relation on $W$, $V:W\longrightarrow 2^\mathcal{P}$ is a valuation function which maps each world to a set of propositional atoms, and $\prec$ is a strict partial order on $W$ that is well-founded (for every $W'\subseteq W$ and every $v\in W'$, either $v$ is $\prec$-minimal in $W'$, or there is a $u\in W'$ that is $\prec$-minimal in $W'$ and $u\prec v$). Satisfaction with respect to $P$ and a world $w\in W$ is defined as follows: For $p\in\mathcal{P}$, $P,w\Vdash p$ iff $p\in V(w)$; $P,w\Vdash \lnot\phi$ iff $P,w\nVdash\phi$; $P,w\Vdash\phi_1\land\phi_2$ iff $P,w\Vdash\phi$ and $P,w\Vdash\phi_2$; $P,w\Vdash \Box_i\phi$ iff $P,v\Vdash\phi$ for every $v\in W$ such that $(w,v)\in R_i$; $P,w\Vdash \dnec_i\phi$ iff $P,v\Vdash\phi$ for every $v\in W$ such that $v\in\min_{\prec}R_i(w)$ (where $R_i(w)=\{w'\mid (w,w')\in R_i\})$; $P,w\Vdash\phi_1\twiddle\phi_2$ whenever $w\notin\min_{\prec}\states{\phi_1}^{P}$ or $w\in\states{\phi_2}^{P}$ (where $\states{\phi_1}^{P}$ refers to those elements $v$ of $W$ for which $P,v\Vdash\phi_1$, and similarly for $\states{\phi_2}^{P}$). Finally, $P\Vdash\phi$ iff $P,w\Vdash\phi$ for every $w\in W$.

Classical multi-modal statements are interpreted in the standard way. Statements of the form ${\dnec}_i\phi$ are true with respect to $P$ and $w$ whenever $\phi$ is true with respect to~$P$ and all the most typical worlds accessible from $w$,  while statements of the form $\dposs_i\phi$ are true with respect to $P$ and $w$ whenever $\phi$ is true with respect to~$P$ and at least one most typical world accessible from $w$. Statements of the form $\phi_1\twiddle\phi_2$ are true in the model~$P$ when $\phi_2$ is true in the most typical $\phi_1$-worlds. 

Britz and Varzinczak present a tableaux method for checking whether or not a statement in $\Lflag$ is satisfiable in some preferential Kripke model. They prove soundness and completeness for their tableaux method, and show that it is \textsc{PSpace}-complete.

\section{Propositional Defeasible Standpoint Logic (PDSL)}\label{section:dmodalities-syntax-semantics}
Having dispensed with the necessary preliminaries, we now proceed to introduce Propositional Defeasible Standpoint Logic, or PDSL.  

\begin{definition}\label{def:language-defeasible-standpoint logic}
    Given a vocabulary $\mathcal{V=(P,S)}$ where $\mathcal{P}$ is a set of propositional atoms and $\mathcal{S}$ is a set of standpoints, we define the set of standpoint expressions $\mathcal{E}$ over $\mathcal{S}$ as 
    \[e::= *\mid s\mid -e \mid e\cap e\]
where $s\in\mathcal{S}$. We define $\LangdSt$ over $\mathcal{V}$ (where $p\in\mathcal{P}$ and $e,d\in\mathcal{E}$) as follows:
    \[\alpha::=\top \mid p\mid e\lesssim d\mid\neg\alpha\mid  \alpha\wedge\alpha\mid \Box_e\alpha\mid{\dnec}_e\alpha\mid \alpha \leadsto \alpha\]
\end{definition}

From this, we can define statements of the form $\alpha\vee\beta$, $\alpha\rightarrow\beta$ and $\alpha\leftrightarrow\beta$ in the usual way. We can also define dual symbols for both classical and defeasible standpoint modalities. That is, we define $\Diamond_e\alpha:=\neg\Box_e\neg \alpha\text{ and }\dposs_e\alpha:=\neg{\dnec}_e\neg \alpha$.
Intuitively, $\Diamond_e\alpha$ reads ``it is possible to $e$ that $\alpha$'' and $\dposs_e\alpha$ represents the stronger notion that ``\textit{in the most typical understandings of $e$'s viewpoint}, it is possible that $\alpha$ holds''. We can also define new standpoint symbols $e\cup d$ and $e\setminus d$ by as $e\cup d:= -(-e\cap-d)\text{ and }e\setminus d:=e\cap-d$,
for all $e,d\in\mathcal{E}$. We are also able to define classical standpoint sharpening statements as
$e\leq d:=\Box_{e\setminus d}\bot$. Note that $e\leq d$ intuitively denotes that ``standpoint $e$ is a more specific version of standpoint $d$''. That is, every precisification associated with $e$'s standpoint can also be associated with $d$'s standpoint. The sentence $e\lesssim d$ can then be thought of as the defeasible version of this sentence, which reads that the \textit{most typical} precisifications associated with standpoint $e$ are also associated with standpoint $d$.


The semantic structure used for defeasible standpoint modalities takes the conventions of the semantics for standpoint propositional logics \cite{alvarezrudolph:propositionalstdpt}, as well as complex standpoint expressions introduced in first-order standpoint logic \cite{alvarez:stdptlogicfocase}, and adds an ordering to precisifications, 
where, intuitively, lower precisifications should be considered as ``more typical'' or more preferred states. This again follows the convention for the defeasible modalities introduced for more generalised multimodal logics \cite{britzvarzin:defeasiblemodalities}. 
\begin{definition}\label{def:standpoint-structure}
    A \textbf{state-preferential} standpoint structure (SPSS) is a quadruple $M=(\Pi,\sigma,\gamma,\prec)$ where,
    \begin{itemize}
        \item $\Pi$ is a set of precisifications.
        \item $\sigma:\mathcal{E}\rightarrow 2^\Pi$ is a function such that $\sigma(*)=\Pi$,   $\sigma(e\cap d)=\sigma(e)\cap \sigma(d)$, and  $\sigma(-e)=\Pi\setminus\sigma(e)$. Moreover, we require that $\sigma(s)\neq \emptyset$ for all $s\in\mathcal{S}$.
        \item $\gamma:\Pi\rightarrow 2^\mathcal{P}$ is a map which assigns a classical valuation to each precisification.
        \item $\prec$ is a strict partial order on $\Pi$ such that for every subset $X$ of $\Pi$, and every $\pi\in X$, either $\pi$ is a $\prec$-minimal element of $X$, or there is a $\pi'\in X$ such that $\pi'$ is a $\prec$-minimal element of $X$, and $\pi'\prec\pi$ (well-foundedness).
    \end{itemize}
\end{definition}

\begin{example}[Example~\ref{Example:WorldSaviours} continued]
Assume $\mathcal{P}=\{\animal,\cheese,\egg\}$ and $\mathcal{S}=\{\Env,\Pcf,\Vga,\Vgt\}$. Figure~\ref{Figure:SPSS} depicts an example of a state preferential standpoint structure for the given vocabulary.
\end{example}

\begin{figure}[h]
\begin{center}
\scalebox{0.85}{
\begin{tikzpicture}
[->,>=stealth',thick, auto,node_style/.style={circle,fill=black,minimum size=0.1mm}]
\node (sigma-Vgt) at (5,5.5) {\textcolor{blue}{\footnotesize{$\sigma(\VgtF)$}}} ;
\node (sigma-Env) at (0.5,4.6) {\textcolor{Green}{\footnotesize{$\sigma(\EnvF)$}}} ;
\node (sigma-Vga) at (5,4.6) {\textcolor{orange}{\footnotesize{$\sigma(\VgaF)$}}} ;
\node (sigma-Vga) at (9.75,2.65) {\textcolor{Purple}{\footnotesize{$\sigma(\PcfF)$}}} ;

\node (pi1) at (1,3) [label=left:\mbox{$\{\eggF\}$}] {$\pi_{1}$} ;
\node (pi2) at (3,4) [label=right:\mbox{$\{\animalF,\cheeseF\}$}] {$\pi_{2}$} ;
\node (pi3) at (3,0) [label=right:\mbox{$\{\cheeseF\}$}] {$\pi_{3}$} ;
\node (pi4) at (6,4) [label=right:\mbox{$\{\animalF,\eggF\}$}] {$\pi_{4}$} ;
\node (pi5) at (6,0) [label=right:\mbox{$\{\eggF\}$}] {$\pi_{5}$} ;
\node (pi6) at (9.5,2) [label=right:\mbox{$\{\cheeseF\}$}] {$\pi_{6}$} ;
\node (pi7) at (9.5,0) [label=right:\mbox{$\{\eggF\}$}] {$\pi_{7}$} ;


\draw [blue,rounded corners,thick] (2,-0.7) rectangle (8.25,5.75) ;
\draw [Green,rounded corners,thick] (-0.75,1.25) rectangle (8,5.25) ;
\draw [orange,rounded corners,thick] (2.25,3) rectangle (7.75,5) ;
\draw (2.25,-0.5) [Purple!70,rounded corners=5pt,thick] -- (2.25,1) -- (8.5,1) -- (8.5,3) -- (11,3) -- (11,-0.5) -- cycle ;

\path
(pi3) edge [dashed] (pi2)
(pi3) edge [dashed] (pi4)
(pi5) edge [dashed] (pi2)
(pi5) edge [dashed] (pi4)
(pi7) edge [dashed] (pi6)
;
\end{tikzpicture}
}
\end{center}
\vspace*{-0.1cm}
\caption{A state preferential standpoint structure for $\mathcal{P}=\{\animal,\cheese,\egg\}$ and $\mathcal{S}=\{\Env,\Pcf,\Vga,\Vgt\}$, where $\Pi=\{\pi_{i} \mid 1\leq i\leq 7\}$, $\sigma(\Env)=\{\pi_{1},\pi_{2},\pi_{4}\}$, $\sigma(\Pcf)=\{\pi_{6}, \pi_{7}\}$, $\sigma(\Vga)=\{\pi_{2},\pi_{4}\}$, and $\sigma(\Vgt)=\{\pi_{2},\pi_{3},\pi_{4},\pi_{5}\}$. Moreover, $\gamma(\pi_{1})=\{\egg\}$, $\gamma(\pi_{2})=\{\animal,\cheese\}$, $\gamma(\pi_{3})=\{\cheese\}$, $\gamma(\pi_{4})=\{\animal,\egg\}$, $\gamma(\pi_{5})=\{\egg\}$, $\gamma(\pi_{6})=\{\cheese\}$, and $\gamma(\pi_{7})=\{\egg\}$. (Standpoints and atomic propositions are abbreviated for conciseness.) The strict partial order on~$\Pi$ is given by $\prec=\{(\pi_{3},\pi_{2}),(\pi_{3},\pi_{4}),(\pi_{5},\pi_{2}),(\pi_{5},\pi_{4}),(\pi_{7},\pi_{6})\}$.}
\label{Figure:SPSS}
\end{figure}

We then define satisfaction for a given SPSS.

\begin{definition}\label{def:model-satisfaction-stdptstructure}
    For an SPSS $M$ and a precisification $\pi\in\Pi$, we define the satisfaction relation $\Vdash$ inductively as follows (where $\alpha,\alpha_1,\alpha_2\in\mathcal{L}_{\mathds{S}}$, $s,s_1,s_2\in\mathcal{S}$ and $p\in \mathcal{P}$):
    \begin{itemize}
        \item $M,\pi\Vdash \top$.
        \item $M,\pi\Vdash p$ iff $p\in\gamma(\pi)$.
        \item $M,\pi\Vdash e\lesssim d$ iff $\min_\prec(\sigma(e))\subseteq \sigma(d)$.
        \item $M,\pi\Vdash \neg \alpha$ iff $M,\pi\nVdash \alpha$.
         \item $M,\pi\Vdash \alpha_1\wedge \alpha_2$ iff $M,\pi\Vdash \alpha_1$ and $M,\pi\Vdash \alpha_2$.
         \item $M,\pi\Vdash \Box_s\alpha$ iff $M,\pi'\Vdash \alpha$ for all $\pi'\in \sigma(s)$.
         \item $M,\pi\Vdash \dnec_s\alpha$ iff $M,\pi'\Vdash \alpha$ for all $\pi'\in \min_\prec(\sigma(s))$.
        \item $M,\pi\Vdash \alpha_1\leadsto\alpha_2$ iff $\pi\notin\min_\prec\llbracket\alpha_1\rrbracket$ or $\pi\in\llbracket\alpha_2\rrbracket$.
         \item $M\Vdash\alpha$ iff $M,\pi\Vdash \alpha$ for all $\pi\in \Pi$.
    \end{itemize}
\end{definition}

We also note here several rules which the semantics introduced above satisfies in general. However, it should be noted that this list is not exhaustive by any means.  Firstly, it should be clear that, since both the language and the semantics introduce notions of defeasibility on top of the existing propositional standpoint logic $\mathcal{L}_\mathbb{S}$, any sentences in $\mathcal{L}_\mathbb{S}$ which are tautologous in the original logic (as discussed by Gómez Álvarez and Rudolph \cite{alvarezrudolph:propositionalstdpt}) are still tautologies in our case.  We therefore will only discuss the defeasible parts of the logic explicitly here. We first compare the defeasible statements to their non-defeasible counterparts.

\begin{proposition}[Supra-classicality]\label{supra-classical-properties-defeasible-symbols}
For any SPSS $M$, and any $\pi\in\Pi$:
    \begin{enumerate}
    \begin{multicols}{2}
        \item $M,\pi\Vdash \
        \Box_e\alpha\implies M,\pi\Vdash \
        \dnec_e\alpha$.
        \item $M,\pi\Vdash \
        \dposs_e\alpha\implies M,\pi\Vdash \
        \Diamond_e\alpha$.
        \item $M,\pi\Vdash \
        e\leq d\implies M,\pi\Vdash \
        e\lesssim d$.
        \item $M,\pi\Vdash \alpha\rightarrow\beta \implies M,\pi\Vdash \alpha\leadsto \beta$.
        \end{multicols}
    \end{enumerate}
\noindent And in general none of the converses hold.
\end{proposition}

This tells us that $\dnec$, $\lesssim$ and $\leadsto$ are strictly weaker notions than their classical counterparts, while $\dposs$ is a stronger notion. We can also note how defeasible modalities affect the unions and intersections of standpoint symbols.

\begin{proposition}\label{proposition:or-and-union-properties-defeasible-modalities}
      For an SPSS $M$ and $\pi\in\Pi$, we have  $M,\pi\Vdash \dnec_{e} \alpha\wedge \dnec_{d} \alpha \implies M,\pi\Vdash \dnec_{e\cup d} \alpha$ and $M,\pi\Vdash \dposs_{e\cup d} \alpha\implies M,\pi\Vdash \dposs_{e} \alpha\vee \dposs_{d} \alpha$. In general, the converses do not hold.
\end{proposition}


  



This shows the relationship that defeasible modalities have when combining them with more ``compound'' standpoint symbols. It is also worth noting that this relationship is different to the case for classical standpoint modalities. For example $M\Vdash \Box_{e} \alpha\wedge \Box_{d} \alpha$ and $M\Vdash \Box_{e\cup d} \alpha$ are equivalent, while this is not the case for defeasible modalities. 

Our semantics also gives us a clearer understanding of $\lesssim$. The statement $e\leq d$ can be written as a modal sharpening $\Box_{e\setminus d}\bot$, which refers to the semantic condition on a standpoint structure $M$ where $\sigma(e)\subseteq \sigma(d)$ \cite{alvarez:stdptlogicfocase}. In the defeasible case, it is clear that an analogous translation does not occur. Consider the SPSS $M$ in Figure \ref{Figure:SPSS}. Clearly,  ${M\Vdash \Vgt\lesssim\Pcf}$, while $M\nVdash \Box_{\Vgt\setminus\Pcf}\bot$ and ${M\nVdash \dnec_{\Vgt\setminus\Pcf}\bot}$, since $\pi_2\in \min_\prec\sigma(\Vgt\setminus\Pcf)$. In fact, the interpretation of the defeasible sharpening $\lesssim$ behaves as a defeasible consequence relation on the standpoint hierarchy, since $s\lesssim t$ can be interpreted semantically as stating ``precisifications in $s$ are \textit{typically} included in $t$''. In order to motivate this, we note that $\lesssim$ satisfies a version of the KLM rationality postulates \cite{kraus:nonmonotonic}.

\begin{proposition}\label{pproposition:lesssim-KLM-properties}
 For $e,d,g\in\mathcal{E}$, an SPSS $M$, and $\pi\in \Pi$ we have:
    \begin{itemize}
        \item $M,\pi\nVdash *\lesssim e\cap - e$\hfill \textit{(Consistency)}
        \item $M,\pi\Vdash e\lesssim e$\hfill \textit{(Reflexivity)}
        \item If $M,\pi\Vdash (e\leq d)\wedge (d\leq e)$ and $M,\pi\Vdash e\lesssim g$ then $M,\pi\Vdash d\lesssim g$ \hfill \textit{(LLE)}
        \item If $M,\pi\Vdash e\lesssim d$, and $M,\pi\Vdash d\leq g$ then $M,\pi\Vdash e\lesssim g$ \hfill \textit{(RW)}
        \item If $M,\pi\Vdash e\lesssim d$, and $M,\pi\Vdash e\lesssim g$ then $M,\pi\Vdash e\lesssim d\cap g$\hfill \textit{(And)}
        \item If $M,\pi\Vdash e\lesssim d$, and $M,\pi\Vdash g\lesssim d$ then $M,\pi\Vdash e\cup g\lesssim d$ \hfill \textit{(Or)}
        \item If $M,\pi\Vdash e\lesssim d$, and $M,\pi\Vdash e\lesssim g$ then $M,\pi\Vdash e\cap g\lesssim d$ \hfill \textit{(CM)}
        \end{itemize}
    \end{proposition}

We can also see that defeasible sharpenings satisfy an adapted version of the classical standpoint axiom (\textbf{P}) \cite{alvarezrudolph:propositionalstdpt}, given by $(e\leq d)\rightarrow(\Box_d \phi\rightarrow\Box_e \phi)$, and thus behaves as a natural extension of the classical sharpening statements. 

\begin{proposition}\label{proposition:adaptedaxiomPfordefeasiblesharpenings}
    For $e,d\in \mathcal{E}$, $\phi\in\LangdSt$, an SPSS $M=(\Pi,\sigma,\gamma,\prec)$, and any $\pi\in\Pi$, we have that $M,\pi\Vdash (e\lesssim d)\rightarrow(\Box_d\phi \rightarrow\dnec_e \phi)$.
\end{proposition}

However, other variants of (\textbf{P}) which incorporate defeasible symbols, such as $({e\lesssim d})\rightarrow(\dnec_d\phi \rightarrow\Box_e \phi)$ or $(e\leq d)\rightarrow(\dnec_d\phi \rightarrow\dnec_e \phi)$, are not satisfied by every SPSS.

It is shown in the more general modal logic $\Lflag$, introduced by Britz and Varzinczak \cite{britzvarzin:defeasiblemodalities}, that the defeasible implication operator $\leadsto$ satisfies a similar set of KLM postulates. This gives us a notion of defeasible implication between Boolean formulas in the logic. However, it is worth noting that, while $\leadsto$ gives us the original KLM-style consequence for Boolean formulas in our language, this intuition does not follow when we combine it with standpoint modalities in our language. In particular, when we bound a defeasible implication $p\leadsto q$ with some defeasible or non-defeasible standpoint modality, we may expect this to tell us something about a standpoints defeasible beliefs. For example, we may expect $\Box_s(p\leadsto q)$ to tell us that ``\textit{from $s$'s standpoint, the most typical instances of $p$ are instances of $q$'}'. The following example shows us that this is not the case. 

\begin{example}
    Consider an SPSS $M=(\Pi,\sigma,\gamma,\prec)$ over propositional atoms $p$ and $q$ and a single standpoint $s$ defined as follows: $\Pi=\{\pi_1,\pi_2\}$, $\sigma(s)=\{\pi_2\}$, $\gamma(\pi_1)=\{p,q\}$, $ \gamma(\pi_2)=\{p\}$, and $\pi_1\prec\pi_2\prec\pi_3$. Then note $M\Vdash \Box_s(p\leadsto q)$ iff $M,\pi\Vdash p\leadsto q$ for all $\pi\in \sigma(s)$. That is, $\pi\in \sigma(s)$ implies $\pi \notin \min_\prec \llbracket p\rrbracket$ or $\pi \in \llbracket q\rrbracket$. Since in the above model $\sigma(s)\cap \min_\prec \llbracket p\rrbracket=\emptyset$, we have $M\Vdash \Box_s(p\leadsto q)$, even though the only precisification in $\sigma(s)$ (and therefore the minimal one) violates $p\rightarrow q$, and so $s$'s standpoint intuitively does not believe that the most typical instances of $p$ are instances of $q$.
\end{example}

Thus, defeasible implication $\LangdSt$ does not serve to represent standpoints holding defeasible beliefs, and a separate semantics is introduced for this problem in \cite{LeisegangRudolphMeyer:SacairStandpoints}. We also note that $\leadsto$ does not provide an intuitive account of defeasible consequences for modal statements. For example,  $M\Vdash \Box_e\alpha\leadsto \beta$ if and only every $\pi\in \min_\prec\llbracket\Box_e\alpha\rrbracket\subseteq\llbracket\beta\rrbracket$. However, it is clear that by definition, either $\llbracket\Box_e\alpha\rrbracket=\emptyset$ or $\llbracket\Box_e\alpha\rrbracket=\Pi$. In the first case, $M\Vdash \Box_e\alpha\leadsto \beta$ trivially; in the second, $M\Vdash \Box_e\alpha\leadsto \beta$ iff $M,\pi\Vdash \beta$ for every $\pi\in\min_\prec \Pi$, which is equivalent to $M\Vdash \dnec_*\beta$. We therefore treat $\leadsto$ as a part of the language which is useful for describing defeasible implication between Boolean statements, but not as an intuitively meaningful statement outside of these bounds.


\section{Satisfiability Checking and Preferential Entailment}\label{section:satisfiability-and-pref-entailment}

In this section we address the notion of preferential satisfiability and preferential entailment in the semantics for $\LangdSt$. We differentiate between local and global satisfaction, as is defined below.

\begin{definition}
    Let $\alpha$ be a sentence in $\LangdSt$. We say that $\alpha$ is \textbf{locally satisfiable} if there exists an SPSS $M=(\Pi,\sigma,\gamma,\prec)$ for which there is some precisification $\pi\in \Pi$ such that $M,\pi\Vdash \alpha$. We say that $\alpha$ is \textbf{globally satisfiable} if there exists an SPSS $M$ such that $M\Vdash \alpha$.  For any finite set $A\subseteq\LangdSt$, we say that $A$ is locally (resp. globally) satisfiable if $\bigwedge A$ is locally (resp. globally) satisfiable.
\end{definition}

Closely related to this is the notion of preferential entailment, which extends the notion of preferential entailment found in the propositional case by Kraus et al. \cite{kraus:nonmonotonic}, and in the case for defeasible modalities in logic \textbf{K} by Britz et al. \cite{britzetal:preferentialreasoningmodal}.

\begin{definition}
    Consider a finite knowledge base $\KB\subseteq \LangdSt$, and a sentence $\alpha\in\LangdSt$. We say that $\KB$ \textbf{preferentially entails} $\alpha$ or write $\KB \vDash_P\alpha$ if, for every SPSS $M$ such that $M\Vdash \phi$ for all $\phi\in\KB$, we have $M\Vdash \alpha$.
\end{definition}

It is noted by Britz et al. \cite{britzmeyervar:normalmodalpreferential} that preferential entailment defined this way induces a monotonic consequence operator. Therefore while the defeasible symbols we introduce are non-monotonic on the object-level, the entailment-level reasoning remains monotonic. The following propositions show that both global satisfiability and preferential entailment can be expressed in terms of local satisfiability.

\begin{proposition}\label{proposition:global-and-entailment-canbe-expressed-as-local}
    Consider a globally satisfiable knowledge base $\KB\subseteq \LangdSt$ and a sentence $\alpha\in\LangdSt$. Then $\alpha$ is globally satisfiable iff $\Box_* \alpha$ is locally satisfiable. Furthermore, $\KB\vDash_P \alpha$ iff $\Box_*(\bigwedge \KB)\wedge\neg \alpha$ is not locally satisfiable.
\end{proposition}

Our first method for checking preferential satisfiability is via translating our logic into classical propositional standpoint logic as described by Gómez Álvarez and Rudolph~\cite{alvarezrudolph:propositionalstdpt}, including complex standpoint expressions. 

\begin{definition}\label{def:defeasible-to-classical-standpoint-translation}
    Let $M=(\Pi,\sigma,\gamma,\prec)$ be a state preferential standpoint structure over the vocabulary $\mathcal{V=(P,S)}$. Then the translation $T(M)=(\Pi,\sigma',\gamma)$ of $M$ is a standpoint structure over the vocabulary $\mathcal{V=(P,S}\cup\tilde{\mathcal{E}})$ where: $\Pi$ and $\gamma$ are the same in each structure, $\tilde{\mathcal{E}}=\{\tilde{e}\mid e\in\mathcal{E}\}$, and $\sigma'$ is defined by $\sigma'(s)=\sigma(s)$ for $s\in\SSS$ and $\sigma'(\tilde{e})=\min_\prec \sigma(e)$ for $\tilde{e}\in\tilde{\mathcal{E}}$. The value of $\sigma'$ is extended to complex standpoint expressions inductively, as is done in the literature \cite{alvarez:stdptlogicfocase}.
\end{definition}

We can then express the satisfiability of defeasible modalities and sharpenings in terms of classical standpoint logic.

\begin{proposition}\label{proposition:translation-to-classical-standpoint-logic}
    For any SPSS $M$, and any sentence $\alpha$ in classical standpoint logic, we have that 
    $M\Vdash \dnec_e\alpha$ iff $T(M)\Vdash \Box_{\tilde{e}}\alpha$,
    and 
    $M\Vdash e\lesssim d$ iff $T(M)\Vdash \tilde{e}\leq d$.
\end{proposition}

This can then give us a means to determine satisfiability, using existing methods in standpoint logics \cite{alvarezrudolph:propositionalstdpt,alvarez:stdptlogicfocase}. 

\begin{corollary}
    For any sentence $\alpha\in \LangdSt$ not containing the symbol ``$\leadsto$'', we have that $\alpha$ is locally satisfiable iff $T(\alpha)$ is satisfiable in classical standpoint semantics, where $T(\alpha)$ is the sentence formed by replacing every instance of $\dnec_e$ occurring in $\alpha$ with $\Box_{\tilde{e}}$ and every subsentence of the form $e\lesssim d$ with $\tilde{e}\leq d$.
\end{corollary}

This approach however, has two weaknesses. Firstly, it does not provide a method for determining the satisfiability of sentences in $\LangdSt$ containing ``$\leadsto$''. Secondly, the set $\tilde{\mathcal{E}}$ as given in Definition \ref{def:defeasible-to-classical-standpoint-translation} is infinite. Since $\mathcal{E}$ behaves as a Boolean algebra of subsets, we could clearly reduce the size of $\tilde{\mathcal{E}}$ to account for equivalences between Boolean formulas. However, for any complex standpoint expression $e$, we cannot in general express $\tilde{e}$ in terms of other, smaller standpoint symbols in $\tilde{\mathcal{E}}$. For example, we cannot express $\widetilde{e\cap d}$ in terms of $\tilde{e}$ and $\tilde{d}$. This means that if $|\SSS|=n$, then the size of $\tilde{\mathcal{E}}$ (once reduced to account for Boolean equivalences) may still be as big as $2^{2^n}$, since we need to consider a new standpoint for every non-equivalent Boolean combination of the original standpoints. This potential double exponential blow-up means that the given translation is not an effective means for determining satisfiability in $\LangdSt$. However, in the restricted case where we only allow sentences with atomic standpoint indexes (i.e., where we restrict $\mathcal{E}$ to the set $\SSS\cup\{*\}$), the size of $\tilde{\mathcal{E}}$ is only $2n+2$. Therefore, since satisfiability for propositional standpoint logic is \textsc{NP}-complete \cite{alvarezrudolph:propositionalstdpt}, the translation above provides an \textsc{NP}-complete means for checking satisfiability in the setting where only atomic standpoint indexes occur in $\alpha$.

However, in the general case for determining satisfiability and preferential entailments, we need to turn to other methods in order to avoid a double exponential blow-up in standpoint symbols. We therefore propose a tableau algorithm for computing whether a given statement in $\LangdSt$ is locally satisfiable. Our tableau is semantic in nature, and follows closely conventions for semantic tableau in related modal logics \cite{britzvarzin:defeasiblemodalities,Castilho1999FormalizingAA,Goré1999}. To this end, we introduce a normal form for sentences in $\LangdSt$.

\begin{definition}
    If a standpoint expression $c\in\mathcal{E}$ is of the form 
    $c=s_1\cap s_2\cap...\cap s_k\cap -s_{k+1}\cap...\cap -s_{l}$
     where $s_1,...,s_l\in\SSS$ we call it a \textbf{standpoint conjunct}. If $k=0$, for the sake of the tableau we add $*$ to the conjunct so that $c=*\cap-(s_{k+1}\cup...\cup s_m)$.

    For any $e\in \mathcal{E}$, we say that $e$ is in \textbf{disjunctive normal form} (DNF) if it is of the form $e=c_1\cup...\cup c_m$
    where $c_1,...,c_m$ are all standpoint conjuncts. We then say a formula $\alpha\in\LangdSt$ is in \textbf{index normal form} (INF) if every standpoint expression which appears in $\phi$ is in disjunctive normal form.
\end{definition}

Since the logic of standpoint expressions operates as a Boolean algebra of subsets, the well-known result that each standpoint expression has an equivalent expression in DNF holds. Therefore, when we check for satisfiability using the following tableau method, we first assume each formula is in INF. In general, in the tableau, an indexed lowercase letter $c$ will denote a standpoint-literal conjunction, while an indexed letter $s$ will denote an atomic standpoint. $e$ and $d$ refer to any standpoint in DNF. This allows us to describe our tableau system.  

\begin{definition}[\cite{britzvarzin:defeasiblemodalities}]
    If $n\in\mathbb{N}$ and $\alpha\in\LangdSt$, then $n::\alpha$ is a \textbf{labelled sentence}.
\end{definition}

Intuitively, the labelled sentence $n::\alpha$ indicates semantically that there is a precisification $n$ in the model such that $\alpha$ holds at $n$. 

\begin{definition}
    A \textbf{skeleton } is a function $\Sigma:\mathcal{E}\rightarrow2^\mathbb{N}$. A \textbf{preference} relation $\prec$ is a binary relation on $\mathbb{N}$.
\end{definition}

A skeleton intuitively assigns each standpoint a set of precisifications, while the preference relation in the tableau acts to construct the preference ordering in an SPSS. Both $\Sigma$ and $\prec$ are built cumulatively, and so at the beginning of the tableau we assume $\prec=\emptyset$ and $\Sigma(e)=\emptyset$ for all $e\in\mathcal{E}$.

\begin{definition}[\cite{britzvarzin:defeasiblemodalities}]
    A branch is a tuple $(\B,\Sigma,\prec)$ where $\B$ is a set of labelled sentences, $\Sigma$ is a skeleton and $\prec$ is a preference relation.
\end{definition}

\begin{definition}[\cite{britzvarzin:defeasiblemodalities}]
    A\textbf{ tableau rule} is of the form
    \[(\rho) \frac{\mathcal{N}:\Gamma}{\mathcal{D}_1;\Gamma_1|...|\mathcal{D}_k;\Gamma_k}\]
where $\mathcal{N}:\Gamma$ is the \textbf{numerator} and $\mathcal{D_1};\Gamma_1|...|\mathcal{D_1};\Gamma_1$ is the \textbf{denominator}.
\end{definition}

As in \cite{britzvarzin:defeasiblemodalities}, $\mathcal{N}$ is a set of labelled sentences called the \textit{main sentences} of $\rho$, while $\Gamma$ specifies a set of conditions on $\Sigma$ and $\prec$. each $\mathcal{D}_i$ is a set of labelled sentences, while each $\Gamma_i$ is a set of conditions that have to be added cumulatively to $\Sigma$ and $\prec$ after the rule is applied. The symbol ``$|$'' indicates where the branch splits. That is, an instance where a non-deterministic choice of possible outcomes has to be explored.

A rule $\rho$ is \textit{applicable} to a branch $(\B,\Sigma,\prec)$ iff $\mathcal{S}$ contains the main sentences of $\rho$ and the conditions of $\Gamma$ are satisfied. The rule \textbf{(non-empty $\mathcal{S}$)} given below has the additional condition that it is only applied when no other rules are applicable.

We also require that \textit{applicable rules} have not already been satisfied. That is, that the denominators have not occurred in the branch before, and in the case of ``fresh'' labels in the denominator, that there are no existing labels $n\in\mathbb{N}$ which satisfy all the conditions of the denominator. We write $n\in W_e$ to denote $n\in\Sigma(e)$ and define $W^\phi_\B:=\{n\in\mathbb{N}\mid n::\phi\in \B\}$, where $\B$ is a branch in the tableau. $n\in\min_{\prec}X$ denotes that $n$ is a minimal element of the set $X$. That is, $n'\in X$ implies $n'\not\prec n$. We also use $n^\star$ to denote the addition of a ``fresh'' label which has not been used before in the tableau. Our tableau rules are defined in Figure \ref{fig:etableau-rules-for-satisfiability}. 

\begin{figure}
    \centering
    \textbf{1. Classical Rules:}
  \[(\bot)\ \frac{n::\alpha,\ n::\lnot\alpha}{n::\bot}\ \ \ \ \ \ (\lnot)\ \frac{n::\lnot\lnot\alpha}{n::\alpha}\ \ \ \ \ \ 
(\land)\ \frac{n::\alpha\land\beta}{n::\alpha,\ n::\beta} \ \ \ \ \ (\lor)\ \frac{n::\lnot(\alpha\land\beta)}{n::\lnot\alpha\ |\ n::\lnot\beta}\ \ \ \ \]
\hspace{2pt}
\textbf{2. Standpoint Hierarchy and Modality Rules:}

\[(\cap)\frac{n\in W_{e\cap d}}{n\in W_e, n\in W_d} \ \ \ \ (\cup)\frac{n\in W_{e\cup d}}{n\in W_{e}| n\in W_{d}}\ \ \ \  (\bot_{-})\frac{n\in W_e, n\in W_{-e}}{n::\bot} \] 

\[ (*_1)\frac{n:: \alpha}{n\in W_*} \ \ \ \ (*_2)\frac{n\in W_e}{n\in W_*}\ \ \ \ (\Diamond_e)\frac{n::\neg\Box_e \alpha}{n^\star::\neg \alpha; n^\star\in W_e}\]

\[(\Box_{e\cup d})\frac{n::\Box_{e\cup d}\alpha}{n::\Box_e \alpha,n::\Box_d \alpha}\ \ \ \ \ (\Box_c)\frac{n::\Box_{s_1\cap s_2...\cap s_k\cap -s_{k+1}\cap...\cap -s_{m}}\alpha; n'\in\Gamma^c}{n'\in W_{s_{k+1}\cup...\cup s_{m}}| n'::\alpha}\]

\[(\Box_{c^+})\frac{n::\Box_{s_1\cap s_2...\cap s_k}\alpha; n'\in \Gamma^c}{ n'::\alpha}\ \ \ \text{ where }\Gamma^c= \{n'\in\mathbb{N}\mid n'\in W_{s_1},...,n'\in W_{s_k}\}.\]

\textbf{3. Defeasibility and Minimality Rules:}
\[(\dposs_e)\frac{n::\neg\dnec_e\alpha}{n^\star:: \neg\alpha; n^\star\in \min_{\prec}W_e}\ \ \ \ (\leadsto)\frac{n::\alpha\leadsto\beta}{n::\neg \alpha|n^\star::\alpha; n^\star\prec n| n::\beta} \ \ \]

\[ (\not\leadsto)\frac{n::\neg(\alpha\leadsto \beta)}{n::\alpha, n::\neg \beta; n\in \min_\prec W^\alpha_\B}\ \ \]

\[ (\bot_\prec)\frac{n\in\min_\prec W, n'\prec n, n'\in W}{n::\bot}\ \ \ \ (\lesssim)\frac{n::c_1\cup...\cup c_m\lesssim d; n'\in \Gamma}{n'\in W_{s_{k+1}\cup...\cup s_l} | n'\in W_d| n^\star \in W_d, \Gamma^\star}\] 

\[(\lesssim^+)\frac{n::c_1\cup...\cup c_m\lesssim d; n'\in \Gamma^+}{ n'\in W_d| n^\star \in W_d, \Gamma^\star}\ 
\ \ \ (\not\lesssim)\frac{n::\neg(e\lesssim d)}{n^\star\in \min_\prec W_e, n^\star\in W_{-d}}\ \ \ \  \]

\[(\dnec)\frac{n::\dnec_{c_1\cup...\cup c_m}\alpha; n'\in \Gamma}{n'::\alpha | n'\in W_{s_{k+1}\cup...\cup s_l}|n^\star ::\alpha; \Gamma^\star}\ \ \ \ (\dnec^+)\frac{n::\dnec_{c_1\cup...\cup c_m}\alpha; n'\in \Gamma^+}{n'::\alpha | n^\star ::\alpha; \Gamma^\star}.\]

where $\Gamma^+=\{n'\in\mathbb{N}\mid n'\in W_{s_1},...,n'\in W_{s_k},  c_i=s_1\cap...\cap  s_k \text{ for some } 1\leq i\leq m\}$, 

$\Gamma^\star=\{n^\star \prec n', n^\star \in W_{c_1\cup... \cup c_m}\}$, and 

$\Gamma=\{n'\in \mathbb{N}\mid n'\in W_{s_1},...,n'\in W_{s_k}, c_i=s_1\cap...\cap  s_k\cap -s_{k+1}\cap...\cap - s_l \text{ where } 1\leq i\leq m\}$\newline

\textbf{(\textit{non-empty} $\mathcal{S}$): }

If, after all other applicable rules are applied to a branch $(\B,\Sigma,\prec)$, there is some $s\in\mathcal{S}$ such that $n\in W_s$ does not appear for any $n\in \mathbb{N}$, then add $n^\star\in W_s$ to $\Sigma$.
    \caption{Tableau Rules for Local Satisfiability in $\LangdSt$}
    \label{fig:etableau-rules-for-satisfiability}
\end{figure}

Many of the rules are straightforward, or have been discussed in the literature \cite{britzvarzin:defeasiblemodalities}, but we provide an explanation for some of the rules which are specific to our case. In Section~2. of Figure~\ref{fig:etableau-rules-for-satisfiability}, the rules $(\cap)$ and $(\cup)$ allocate every label $n$ within a  complex standpoint into the possible set of standpoint literals associated with this complex standpoint literals. This is important to ensure termination, and in order for the modal rules to be applied to each label correctly, since the conditions in the numerators of modal rules are written specifically in terms of atomic standpoints. Rules $(*_1)$ and $(*_2)$ make sure every label is associated to the universal standpoint. Rule $(\bot_{-})$ accounts for the semantic contradiction which occurs when a label is allocated to both a standpoint and its negation. Rules $(\Box_c)$ and $(\Box_c^+)$ deal with sentences bound by strict modalities with conjuncts as indexes, and the conditions are phrased so that the applicability of the rule applies to labels which fulfil the required atomic standpoints. $(\Box_c^+)$ deals with conjuncts with no negative literals, while $(\Box_c)$ deals with conjuncts with negative literals. The first branch in the denominator of $(\Box_c)$ intuitively accounts for the case where a label $n$ satisfies the positive literals in the conjunct but is not in one of the negative conjuncts, and so the sentence $n::\alpha$ need not appear.
Rule $(\Box_{e\cup d})$ deals with indexes which have unions of multiple conjuncts. Section 3. describes the behaviour of defeasible parts of the logic. Rules $(\dnec_e)$, $(\leadsto)$, $(\not\leadsto)$ and $(\bot_\prec)$ follow similar intuitions to those appearing in tableau construted by Britz and Varzicnzak \cite{britzvarzin:defeasiblemodalities}. Rules $(\lesssim)$ and $(\lesssim^+)$ allocates necessary labels to a new standpoint when a defeasible sharpening occurs. The conditions are again referred to in terms of atomic standpoints, and rules $(\lesssim)$ and $(\lesssim^+)$ differ in the same manner that $(\Box_c)$ and $(\Box_c^+)$ do. Both rules also have a branch which expresses that when $n::e\lesssim d$ occurs, it is possible that a label in $W_e$ is not allocated to $W_d$ on the account of it being non-minimal. Rule $(\not\lesssim)$ deals with negated defeasible sharpenings. Lastly, rules $(\dnec)$ and $(\dnec^+)$ deal with sentences bound by $\dnec$ with a similar differentiation between conjuncts with literals and without. It is worth noting that by Proposition \ref{proposition:or-and-union-properties-defeasible-modalities}, we cannot break unions of conjuncts down into their parts as is done in classical modalities, and so have to treat the general case of sentences in INF.

\begin{definition}\cite{britzvarzin:defeasiblemodalities}\label{definition:tableau}
    A\textbf{ tableau} $\T$ for $\alpha\in\LangdSt$ is the limit of a sequence $\T^0,...,\T^k,...$ of sets of branches where the initial $\T^0:=\{(\{0::\alpha\},\emptyset,\emptyset)\}$ and every $\T^{i+1}$ is obtained by the application of one of the \textit{applicable} rules in Figure \ref{fig:etableau-rules-for-satisfiability} to some branch in $\T^{i}$. Such a limit is denoted $\T^\infty$.
\end{definition}

We assume here that the limit is only found once every applicable rule is applied. We say  a tableau is \textit{saturated}
if no rule is applicable to any of its branches.

\begin{definition}\label{def:closed-and-open-tableau}
    A branch $(\B,\Sigma,\prec)$ is \textbf{closed }iff $n::\bot \in \B$ for some $n$. A saturated tableau $\T$ for $\alpha\in\LangdSt$ is closed iff all its branches are closed. If a saturated tableau $\T$ is not closed, we say that it is \textbf{open}.
\end{definition}

We then can describe the \textit{tableau algorithm} for local satisfaction checking as follows: If we are given a sentence $\alpha\in\LangdSt$ and want to check if it is locally satisfiable, then we construct a saturated tableau for $\alpha$ as in Definition \ref{definition:tableau}. If the resulting tableau is closed, then we conclude that $\alpha$ \textit{is not} locally satisfiable, and if the tableau is open, we conclude that $\alpha$ \textit{is} locally satisfiable. However, in order for this algorithm to be useful, we need the proposed tableau calculus to be sound, complete and to terminate. We therefore present the following two theorems, which are the main results of our paper:

\begin{theorem}[Complexity] \label{therorem:tableau-in-pspace}
    The tableau algorithm runs in \textsc{PSpace}.
\end{theorem}

This tells us that our tableau calculus terminates and that it is in the same complexity class as the tableau algorithms for modal logic $\textbf{K}_n$ with defeasible modalities \cite{britzmeyervar:normalmodalpreferential}, as well as for the classical normal modal logics \textbf{K} and $\textbf{K}_n$ \cite{ladner:complexity-of-K,Halpern:complexity-for-modal-logics}.

\begin{theorem}[Soundness and Completeness.]\label{theorem:soundness and completeness}
   The tableau algorithm is sound and complete with respect to local satisfiability in SPSS semantics.
\end{theorem}

 Moreover, by Theorem \ref{theorem:soundness and completeness} and Proposition \ref{proposition:global-and-entailment-canbe-expressed-as-local}, we can easily adapt our tableau algorithm in order to obtain algorithms for global satisfiability and for preferential entailment which are sound, complete and computable in \textsc{PSpace}. Lastly, it is worth noting that the rules in Sections 1 and 2 in Figure \ref{fig:etableau-rules-for-satisfiability} provide a tableau calculus for an extension of Gómez Álvarez and Rudolph's \cite{alvarezrudolph:propositionalstdpt} classical propositional standpoint logic, in which full complex standpoint expressions are permitted. This is on its own a novel contribution to the field of standpoint logics.



\section{Related Work}\label{Sec:RelatedWork}

The most closely related work to this paper is that of Leisegang et al. \cite{LeisegangRudolphMeyer:SacairStandpoints} who also consider combining defeasibility and standpoint logics. However, their paper aims at representing situations where standpoints \textit{hold} defeasible beliefs, while this paper considers defeasibility and typicality relations between the precisifications themselves. In particular, the language DRSL is given in the form $\phi::=\psi\mid\phi\wedge\phi\mid\Box_s\psi\mid\Diamond_s\psi$,\footnote{As well as classical sharpening statements.} where $\psi$ is a boolean formula or a defeasible implication $\alpha\leadsto \beta$ where $\alpha$ and $\beta$ are Boolean. The semantics is given by \textit{ranked standpoint structures}, which consist of a triple $M=(\Pi,\sigma,\gamma)$ where $\Pi$ and $\sigma$ are as in classical standpoint structures. However, $\gamma$ maps each precisification not to a classical valuation, but a \textit{ranked interpretation} as defined by Lehmann and Magidor \cite{lehmann:conditionalentail}. That is, each $\gamma(\pi)$ is a ranking function $\gamma(\pi):2^\mathcal{P}\rightarrow \mathbb{N}\cup\{\infty\}$. Such a ranking function intuitively expresses how ``typical'' or preferred a state of the world is, with lower rankings signifying more typical states. This in turn induces a preference ordering for each precisification where $v\prec_{\pi} u$ iff $\gamma(\pi)(v)<\gamma(\pi)(u)$. Then if $\alpha$ and $\beta$ are Boolean formulae, $M,\pi\Vdash \alpha$ iff $\gamma(v)\neq \infty$ implies $v\Vdash \alpha$ $M,\pi\Vdash \alpha\leadsto \beta$ iff $\min_{\prec_\pi}\llbracket\alpha\rrbracket\subseteq \llbracket\beta\rrbracket$. The satisfaction standpoint modalities and conjunctions are defined inductively on top of this as in the classical case (for example, $M,\pi\Vdash \Box_s\phi$ iff $M,\pi'\Vdash \phi$ for all $\phi\in\sigma(s)$). These semantics for DRSL therefore utilize preference orderings, but such orderings are internal to the underlying valuation of each precisification, while the ordering in PDSL occurs on the set of precisifications itself. Another distinction in the work of Leisegang et al. \cite{LeisegangRudolphMeyer:SacairStandpoints}  is that it focuses on extending \textit{rational closure}, a non-monotonic form of reasoning introduced by Lehman and Magidor \cite{lehmann:conditionalentail}, into the DRSL case, while our work focusses on an extension of preferential entailment.

Besides this work, defeasibility has been considered for basic normal modal logics by Britz et al. \cite{britzetal:preferentialreasoningmodal,britzmeyervar:normalmodalpreferential} and Britz and Varzinczak \cite{britzvarzin:defeasiblemodalities}, and in the case of Linear Temporal Logic by Chafik et al. \cite{chafik:defeasiblelineartemporal}. Standpoint modalities have been considered in the propositional case by Gómez Álvarez and Rudolph \cite{alvarezrudolph:propositionalstdpt}, in first-order logic and its decidable fragments by Gómez Álvarez  et al. \cite{alvarez:stdptlogicfocase,alvarezrudolphstrass:standpointEL,alvarezstrassrudolph:standpointELplus} and in the case of linear temporal logic by Gigante et al. \cite{alvarezlyon:stndpttemporal}, who also use semantic tableau methods. Other forms of non-monotonic reasoning in standpoint logics have been considered by Gorczyca and Straß \cite{gorczyca-strass:nmonotonic-standpoint-s4f}.

\section{Conclusion}\label{Sec:Conclusion}

In this paper, we propose an extension to both defeasible modalities, and standpoint logics by considering a logic of defeasible standpoint modalities. We define the language $\LangdSt$ which extends propositional standpoint logic with defeasible modalities of the form $\dnec_e$ and $\dposs_e$, as well as defeasible standpoint sharpenings and implications. The main contribution of the paper is to provide a semantics for the logical language in Section \ref{section:dmodalities-syntax-semantics}, and provide a sound, complete and terminating method to check satisfiability with respect to these semantics. In particular, in Section \ref{section:satisfiability-and-pref-entailment} we consider a translation to plain standpoint logic which allows for an \textsc{NP}-complete satisfiability checking in a restricted setting, and go on to provide a tableau algorithm for the unrestricted case which is computable in \textsc{PSpace}.

For future work, we believe investigating other forms of non-monotonic entailment for defeasible reasoning in the propositional case \cite{lehmann:conditionalentail,lehmann:lexicographicreason} in the language $\LangdSt$ would expand the given logic's ability to reason prototypically about information. Moreover, it would be worth investigating whether defeasible standpoint modalities can be added to more expressive logics, such as lightweight description logics, where classical standpoint modalities have been investigated \cite{alvarezrudolphstrass:standpointEL,alvarezstrassrudolph:standpointELplus}. Lastly, it would be worth investigating a deeper comparison and fusion of the logic proposed in our paper with the related approach to defeasibility in standpoint logics given by Leisegang et al. \cite{LeisegangRudolphMeyer:SacairStandpoints}.

\section*{Acknowledgements}
This work is based on the research supported in part by the National Research Foundation of South Africa (REFERENCE NO: SAI240823262612). This work has been partially supported by the Programme Hubert Curien Campus France Protea 48957ZC ``Symbolic Artificial Intelligence for Cyber Security''. Nicholas Leisegang is supported by the University of Cape Town Science Faculty PhD Fellowship, and the South African Media, Information and Communication Technologies Sector Education and Training Authority (MICT-SETA) Bursary.

\bibliographystyle{splncs04}
\bibliography{defeasiblestandpointmodalities}

\newpage
\appendix
\section{Appendix}
\subsection{Proofs for Section \ref{section:dmodalities-syntax-semantics}}

\noindent\textbf{Proposition \ref{supra-classical-properties-defeasible-symbols}} \textit{For any state-preferential standpoint structure $M$, and any precisification $\pi$ the following holds:}
    \begin{enumerate}
        \item $M,\pi\Vdash \
        \Box_e\alpha\implies M,\pi\Vdash \
        \dnec_e\alpha$.
        \item $M,\pi\Vdash \
        \dposs_e\alpha\implies M\Vdash \
        \Diamond_e\alpha$.
        \item $M,\pi\Vdash \
        e\leq d\implies M,\pi\Vdash \
        e\lesssim d$.
        \item $M,\pi\Vdash \alpha\rightarrow\beta \implies M,\pi\Vdash \alpha\leadsto \beta$.
    \end{enumerate}
\textit{And in general none of the converses hold.}

\begin{proof}
Let $M=(\Pi,\sigma,\gamma,\prec)$ be an arbitrary SPSS. We first prove the positive statements:
    \begin{enumerate}
        \item If $M,\pi\Vdash \
        \Box_e\alpha$, then $M,\pi'\Vdash \alpha$ for all $\pi'\in \sigma(e)$ then for each $\pi''\in\min_\prec \sigma(e)\subseteq \sigma(e)$ we have $M,\pi''\Vdash \alpha$. Hence $M,\pi\Vdash\
        \dnec_e\alpha$.
        \item If $M,\pi\Vdash \
        \dposs_e\alpha$, then there exists $\pi'\in\min_\prec \sigma(e)$ such that $M,\pi'\Vdash \alpha$. But then, since $\min_\prec \sigma(e)\subseteq \sigma(e)$,  we have $M,\pi\Vdash \Diamond_e \alpha$.
        \item $M,\pi\Vdash \
        e\leq d$ iff. $\sigma(e)\subseteq \sigma(d)$. Thus, $\min_\prec \sigma(e)\subseteq \sigma(e)\subseteq \sigma(d)$ and $M,\pi\Vdash \
        e\lesssim d$.
        \item If $M,\pi\Vdash \alpha\rightarrow \beta=\neg \alpha\vee \beta$ then $M,\pi\Vdash \neg\alpha$ or $M,\pi\Vdash \beta$. Equivalently, $\pi\notin \llbracket \alpha \rrbracket$ or $\pi\in \llbracket \beta \rrbracket$. Then, since $\min_\prec\llbracket \alpha \rrbracket\subseteq \llbracket \alpha \rrbracket$, the above implies $\pi\notin \min_\prec\llbracket \alpha \rrbracket$ or $\pi\in \llbracket \beta \rrbracket$. That is, $M,\pi\Vdash \alpha\twiddle\beta$.
    \end{enumerate}

    In order to show the converses do not hold we consider the vocabulary $\mathcal{P}=\{p,q\}$ and $\SSS=\{s,t\}$. The we define the SPSS $M=(\Pi,\sigma,\gamma,\prec)$ as $\Pi=\{\pi_1,\pi_2,\pi_3\}$, $\sigma(s)=\{\pi_1,\pi_2\}$ and $\sigma(t)=\{\pi_1,\pi_3\}$. Then $\gamma(\pi_1)=\{p\}$, $\gamma(\pi_2)=\{q\}$ and $\gamma(\pi_3)=\{p,q\}$. Lastly, $\prec=\{(\pi_1,\pi_2), (\pi_1,\pi_3)\}$.

    For the first three converses we show that there are sentences for which every precisification in $\Pi$ breaks the reverse implication. Notice that $M,\Vdash \dnec_s p$ since $M,\pi_1\Vdash p$. However, since $M,\pi_2\nVdash p$ we have $M\nVdash \Box_s p$. Moreover, $M\Vdash \Diamond_s q$ since $M,\pi_2\Vdash 1$ but $M\nVdash \dposs_s q$ since $M,\pi_1\nVdash q$. Then consider that $\min_\prec \sigma(s)=\{\pi_1\}\subseteq \sigma(t)$ and so $M\Vdash s\lesssim t$. However, $\pi_2\in \sigma(t)\setminus \sigma(s)$ and so  $M\nVdash s\leq t$. For the converse of 4., $M,\pi_3\Vdash p\leadsto \neg q$ since $\pi_3\notin\min_\prec\llbracket p\rrbracket$. However, $M\nVdash p\rightarrow \neg q$ since $M,\pi_3\Vdash p\wedge q$. 
\end{proof}

\noindent\textbf{Proposition \ref{proposition:or-and-union-properties-defeasible-modalities}.} \textit{For any state-preferential standpoint structure $M$, and any precisification $\pi$ the following holds:}
      \begin{enumerate}
          \item $M,\pi\Vdash \dnec_{e} \alpha\wedge \dnec_{d} \alpha \implies M,\pi\Vdash \dnec_{e\cup d} \alpha$.
 \item $M,\pi\Vdash \dposs_{e\cup d} \alpha\implies M\Vdash \dposs_{e} \alpha\vee \dposs_{d} \alpha$.
      \end{enumerate}
       \noindent \textit{And in general the converses do not hold.}

\begin{proof}
    Let $M=(\Pi,\sigma,\gamma,\prec)$ be an arbitrary SPSS and $\pi\in\Pi$:

\begin{enumerate}
    \item Suppose $M,\pi\Vdash \dnec_{e} \alpha\wedge \dnec_{d} \alpha$, and (in order to obtain a contradiction) that $M,\pi\nVdash {\dnec_{e\cup d}} \alpha$. Then there exists some, $\pi'\in\min_\prec \sigma(s\cup d)$ such that $M,\pi'\nVdash \alpha$. However, since $\pi'\in\sigma(s\cup d)$, then either $\pi'\in \sigma(e)$ or $\pi'\in \sigma(d)$. Without loss of generality, assume $\pi'\in \sigma(e)$. Since $M,\pi\Vdash \dnec_{e} \alpha$, and $M,\pi'\nVdash \alpha$ then $\pi'$ cannot be minimal in $\sigma(s)$. That is, there is some  $\pi''\in\sigma(s)$ such that $\pi''\prec \pi'$. However, $\pi''\in \sigma(s)\cup \sigma(d)$ and so then $\pi'$ cannot be minimal in $\sigma(s\cup d)$, which gives us a contradiction. 
    \item Notice that the equivalent converse to statement 2. is $M,\pi\nVdash \dposs_{e} \alpha\vee \dposs_{d} \alpha \implies M,\pi\nVdash \dposs_{e\cup d} \alpha$. Then note that $M,\pi\nVdash \dposs_{e} \alpha\vee \dposs_{d} \alpha $ iff $M,\pi\Vdash \neg(\neg\dnec_{e} \neg\alpha\vee \neg\dnec_{d}\neg \alpha )=\dnec_{e} \neg\alpha\wedge \dnec_{d} \neg\alpha$. On the other hand  $M,\pi\nVdash \dposs_{e\cup d} \alpha$ iff$M,\pi\Vdash \neg(\neg\dnec_{e\cup d} \neg\alpha)=\dnec_{e\cup d} \neg\alpha$. That is, the converse of 2. is an instance of case 1. and so it follows from the previous point in the proof.
\end{enumerate}

For the converses, we show contradictions in the stronger case as in the proof of Proposition \ref{supra-classical-properties-defeasible-symbols}. In order to see that that the converses do not hold consider the vocabulary $\mathcal{P}=\{p\}$ and $\SSS=\{s,t\}$, and the SPSS $M=(\Pi,\sigma,\gamma,\prec)$ defined as follows. $\Pi=\{\pi_1,\pi_2,\pi_3\}$, $\sigma(s)=\{\pi_1,\pi_2\}$, $\sigma(t)=\{\pi_1,\pi_3\}$, $\gamma(\pi_2)=\emptyset$, $\gamma(\pi_1)=\gamma(\pi_3)=\{p\}$ and lastly $\prec=\{(\pi_3,\pi_2)\}$. Then note that $\sigma(s\cup t)=\Pi$ and so $\min_\prec \sigma(s\cup t)=\{\pi_1,\pi_3\}$. Then $M,\pi_1\Vdash p$ and $M,\pi_3\Vdash p$ and so $M,\Vdash\dnec_{s\cup t} p$. However, $\min_\prec \sigma(s)=\{\pi_1,\pi_2\}$ and $M,\pi_2\nVdash p$. Therefore $M\nVdash \dnec_s p$. Moreover, since $M,\pi_2\Vdash \neg p$ we have $M\Vdash \dposs_s \neg p$ and thus $M\Vdash \dposs_s \neg p\vee \dposs_t \neg p$. However, since $\pi_2\notin \min_\prec \sigma(s\cup t)$ it also follows that $M\nVdash \dposs_{s\cup t} \neg p$.
\end{proof}

\noindent\textbf{Proposition \ref{pproposition:lesssim-KLM-properties}.}\textit{ For any standpoint expressions $e,d,g\in\mathcal{E}$, any state-preferential standpoint structure $M$, and any precisification $\pi$ we have:}
    \begin{itemize}
        \item $M,\pi\nVdash *\lesssim e\cap - e$\hfill \textit{(Consistency)}
        \item $M,\pi\Vdash e\lesssim e$\hfill \textit{(Reflexivity)}
        \item \textit{If $M,\pi\Vdash (e\leq d)\wedge (d\leq e)$ and $M,\pi\Vdash e\lesssim g$ then $M,\pi\Vdash d\lesssim g$} \hfill \textit{(Left Logical Equivalence)}
        \item \textit{If $M,\pi\Vdash e\lesssim d$, and $M\pi\Vdash d\leq g$ then $M,\pi\Vdash e\lesssim g$} \hfill \textit{(Right Weakening)}
        \item \textit{If $M,\pi\Vdash e\lesssim d$, and $M,\pi\Vdash e\lesssim g$ then $M,\pi\Vdash e\lesssim d\cap g$} \hfill \textit{(And)}
        \item \textit{If $M,\pi\Vdash e\lesssim d$, and $M,\pi\Vdash g\lesssim d$ then $M,\pi\Vdash e\cup g\lesssim d$ \hfill \textit{(Or)}}
        \item \textit{If $M,\pi\Vdash e\lesssim d$, and $M,\pi\Vdash e\lesssim g$ then $M,\pi\Vdash e\cap g\lesssim d$ \hfill \textit{(Cautious Monotonicity)}}
        \end{itemize}

\begin{proof}
        Let $M=(\Pi,\sigma,\gamma,\prec)$ be an arbitrary SPSS, and $\pi\in\Pi$. Than consider $e,d,g\in \mathcal{E}$:
\begin{itemize}
    \item \textit{(Consistency):} Note that  $\sigma(e\cap -e)=\sigma(e)\cap\sigma(-e)=\sigma(e)\cap(\Pi\setminus\sigma(e))=\emptyset$. By definition $\sigma(*)=\Pi$ is not empty and therefore by smoothness, $\min_\prec\sigma(*)\neq \emptyset$. Thus, $\min_\prec\sigma(*)\nsubseteq \emptyset=\sigma(e\cap -e)$ and so $M,\pi\nVdash *\lesssim e\cap - e$.
    \item \textit{(Reflexivity):} By definition $\min_\prec\sigma(e)\subseteq \sigma(e)$ for any standpoint expression, therefore $M,\pi\Vdash e\lesssim e$.
    \item  \textit{(Left Logical Equivalence):} Suppose $M,\pi\Vdash (e\leq d)\wedge (d\leq e)$ and $M,\pi\Vdash e\lesssim g$. Then $\sigma(e)=\sigma(d)$ and $\min_\prec\sigma(e)\subseteq \sigma(g)$. Consequently $\min_\prec\sigma(d)\subseteq \sigma(g)$ and so $M,\pi\Vdash d\lesssim g$.
    \item \textit{(Right Weakening):} Suppose $M,\pi\Vdash e\lesssim d$, and $M,\pi\Vdash d\leq g$. Then $\min_\prec\sigma(e)\subseteq \sigma(d)$ and $\sigma(d)\subseteq \sigma(g)$ and so $\min_\prec\sigma(e)\subseteq \sigma(g)$. Hence $M,\pi\Vdash e\lesssim g$.
    \item\textit{(And):} Suppose $M,\pi\Vdash e\lesssim d$, and $M,\pi\Vdash e\lesssim g$. Then $\min_\prec\sigma(e)\subseteq \sigma(d)$  and $\min_\prec\sigma(e)\subseteq \sigma(g)$. Therefore $\min_\prec\sigma(e)\subseteq \sigma(d)\cap \sigma(g)=\sigma(d\cap g)$ and so  $M,\pi\Vdash e\lesssim d\cap g$.
    \item \textit{(Or):} Suppose $M,\pi\Vdash e\lesssim d$, and $M,\pi\Vdash g\lesssim d$. Then, $\min_\prec\sigma(e)\subseteq \sigma(d)$ and $\min_\prec\sigma(g)\subseteq \sigma(d)$. Then if $\pi'\in\min_\prec\sigma(e\cup g)$ we have $\pi'\in \sigma(e)\cup \sigma(g)$. Without loss of generality assume $\pi'\in\sigma(e)$. Then if $\pi'\notin \min_\prec\sigma(e)$, there exists $\pi''\in\sigma(e)\subseteq \sigma(e\cup g)$ such that $\pi''\prec \pi'$. However, then $\pi'$ is not minimal in  $\sigma(e\cup g)$ and we have a contradiction. Therefore, $\pi'\in \min_\prec\sigma(e)$ and so $\pi'\in \sigma(d)$. That is, $\min_\prec\sigma(e\cup g)\subseteq \sigma(d)$ and it follows that $M,\pi\Vdash e\cup g\lesssim d$.
    \item \textit{(Cautious Monotonicity):} Suppose $M,\pi\Vdash e\lesssim d$, and $M,\pi\Vdash e\lesssim g$. Then $\min_\prec\sigma(e)\subseteq \sigma(d)$ and $\min_\prec\sigma(e)\subseteq \sigma(g)$. Consider $\pi'\in \min_\prec(e\cap g)$. Then $\pi'\in \sigma(e)\cap \sigma(g)$. Suppose then that $\pi'\notin \min_\prec \sigma(e)$. Then there exists $\pi''\in\min_\prec \sigma(e)$ such that $\pi''\prec \pi$. But by assumption $\pi''\in \sigma(g)$ and so $\pi''\in \sigma(e \cap g)$, which contradicts the assumption that $\pi'\in \min_\prec(e\cap g)$. Therefore, if  $\pi'\in \min_\prec(e\cap g)$ we must have $\pi'\in\min_\prec \sigma(e)\subseteq \sigma(d)$. Therefore $\min_\prec \sigma(e\cap g)\subseteq \sigma(d)$ and so $M,\pi\Vdash e\cap g\lesssim d$.
\end{itemize}
\end{proof}

\noindent\textbf{Proposition \ref{proposition:adaptedaxiomPfordefeasiblesharpenings}} \textit{    For $e,d\in \mathcal{E}$, $\phi\in\LangdSt$, an SPSS $M=(\Pi,\sigma,\gamma,\prec)$ and any $\pi\in\Pi$, we have that $M,\pi\Vdash (e\lesssim d)\rightarrow(\Box_d\phi \rightarrow\dnec_e \phi)$.}

\begin{proof}
    Suppose, in order to obtain a contradiction, that there is some SPSS $M=(\Pi,\sigma,\gamma,\prec)$ and $\pi\in\Pi$ such that $M,\pi\nVdash (e\lesssim d)\rightarrow(\Box_d\phi \rightarrow\dnec_e \phi)$. Equivalently $M,\pi\Vdash \neg((e\lesssim d)\rightarrow(\Box_d\phi \rightarrow\dnec_e \phi))$, which using standard De Morgan laws can be reduced to $M,\pi\Vdash (e\lesssim d)\wedge\Box_d\phi\wedge \neg \dnec_e \phi$. Since $M,\pi\Vdash\neg \dnec_e \phi$ (equivalently $M,\pi\Vdash \dposs_e \neg\phi$), there exists some $\pi'\in\min_\prec\sigma(e)$ such that $M,\pi'\Vdash \neg \phi$. Then since $M,\pi\Vdash (e\lesssim d)$ we have $\min_\prec\sigma(e)\subseteq \sigma(d)$, and in particular $\pi'\in \sigma(d)$. Lastly, since $M,\pi\Vdash\Box_d\phi$ we have $M,\pi''\Vdash \phi$ for all $\pi''\in\sigma(d)$. But then $M,\pi'\Vdash \phi$ which is a contradiction.
\end{proof}

\subsection{Proofs for Section \ref{section:satisfiability-and-pref-entailment}}

\textit{\textbf{Proposition \ref{proposition:global-and-entailment-canbe-expressed-as-local}.}   For any globally satisfiable knowledge base $\KB\subseteq \LangdSt$ and sentence $\alpha\in\LangdSt$:
    \begin{enumerate}
        \item $\alpha$ is globally satisfiable iff $\Box_* \alpha$ is locally satisfiable.
        \item $\KB\vDash_P \alpha$ iff $\Box_*(\bigwedge \KB)\wedge\neg \alpha$ is not locally satisfiable.
    \end{enumerate}}

    \begin{proof}
        \begin{enumerate}
            \item $(\implies)$: If $\alpha$ is globally satisfiable, then there exists $M=(\Pi,\sigma,\gamma,\prec)$ such that $M\Vdash \alpha$. Equivalently, $M,\pi\Vdash\alpha$ for all $\pi\in\Pi=\sigma(*)$. Therefore, $M,\pi'\Vdash\alpha$ for every $\pi'\in \sigma(*)$, meaning $M\Vdash \Box_*\alpha$ by definition. Then, picking any random $\pi''\in\Pi$, we have that  $M,\pi''\Vdash \Box_*\alpha$ and so $\Box_*\alpha$ is locally satisfiable.

            $(\impliedby)$: If $\Box_*\alpha$ is locally satisfiable we have that there exists $M=(\Pi,\sigma,\gamma,\prec)$ and $\pi\in \Pi$ such that $M,\pi\Vdash\Box_*\alpha$. By definition, $M,\pi'\Vdash\Box_*\alpha$ for all $\pi'\in\sigma(*)=\Pi$, and so $\alpha$ is globally satisfiable.
            
            \item $(\implies)$: If $\KB\vDash_P \alpha$ we have that whenever $M'\Vdash \phi$ for all $\phi\in\KB$ then $M'\Vdash \alpha$. Suppose then that $\Box_*(\bigwedge \KB)\wedge\neg \alpha$ is locally satisfiable. Then there exists $M=(\Pi,\sigma,\gamma,\prec)$ and $\pi\in \Pi$ such that $M,\pi\Vdash\Box_*(\bigwedge \KB)\wedge\neg \alpha$. Equivalently $M,\pi\Vdash\Box_*(\bigwedge \KB)$ and $M,\pi\Vdash\neg\alpha$. We then have  $M,\pi\nVDash\alpha$. But then, by point 1. $M\Vdash\bigwedge \KB$, and since $\KB\vDash_P \alpha$, we must have that $M\Vdash \alpha$ (i.e., $M,\pi'\Vdash \alpha$ for all $\pi'\in\Pi$). In particular, $M,\pi\Vdash \alpha$ which is a contradiction. Therefore, if $\KB\vDash_P \alpha$ then $\Box_*(\bigwedge \KB)\wedge\neg \alpha$ is not locally satisfiable.

         $(\impliedby)$: If $\Box_*(\bigwedge \KB)\wedge\neg \alpha$ is not locally satisfiable, then there is no $M$ such that $M,\pi\Vdash \Box_*(\bigwedge \KB)\wedge\neg \alpha$ for some $\pi\in\Pi$. Moreover, since $\KB$ is globally satisfiable there exists some $M'=(\Pi',\sigma',\gamma',\prec')$ such that $M'\Vdash \bigwedge \KB$. Then, by 1. for every $\pi'\in \Pi'$, we have $M',\pi'\Vdash \Box_*(\bigwedge \KB)$. By our assumption we then cannot have $M',\pi'\Vdash \neg \alpha$ and so  $M',\pi'\Vdash \alpha$ for all $\pi'\in\Pi'$. That is,  $M'\Vdash \alpha$. Lastly, since $M'$ was an arbitrary SPSS such that $M'\Vdash \bigwedge \KB$, we have that $M''\Vdash \alpha$ for any other $M''$ such that $M'\Vdash \bigwedge \KB$. Therefore $\KB\vDash_p \alpha$.
        \end{enumerate}
    \end{proof}

\noindent\textbf{Proposition \ref{proposition:translation-to-classical-standpoint-logic}.}\textit{ For any SPSS $M$, and any sentence $\alpha$ in classical standpoint logic, we have that 
    $M\Vdash \dnec_e\alpha$ iff $T(M)\Vdash \Box_{\tilde{e}}\alpha$,
    and 
    $M\Vdash e\lesssim d$ iff $T(M)\Vdash \tilde{e} \leq d$.}

\begin{proof}
    For the first  equivalence, note that  $M\Vdash \dnec_e\alpha$ iif and only if
         $M, \pi \Vdash \alpha$ for all $\pi\in\min_\prec(\sigma(e))$. Equivalently, $T(M),\pi\Vdash \alpha$  for all $\pi\in\sigma'(\tilde{e})$. Note here that since $\alpha$ is a classical standpoint sentence and $\Pi$ and $\gamma$ are the same in both structures, the classical standpoint logic statements satisfied by $M$ at $\pi$ are the same as those satisfied by $T(M)$ at $\pi$. Lastly, by the definition of satisfaction in standpoint structures, the previous statement is equivalent to $T(M)\Vdash \Box_{\tilde{e}}\alpha$.

         For the second equivalence, we note that $M\Vdash e\lesssim d$ iff $\min_\prec\sigma(e)\subseteq \sigma(d)$. In the semantic structure of $T(M)$, this is equivalent to $\sigma'(\tilde{e})\subseteq \sigma'(d)$. That is, $T(M)\Vdash \tilde{e}\leq d$.
\end{proof}

    \noindent\textbf{Complexity and Termination Proof:}\\

    \noindent\textbf{Theorem \ref{therorem:tableau-in-pspace} (Complexity).}
   \textit{ The tableau algorithm runs in \textsc{PSpace}.}

\begin{proof}
    We use the well-known result \textsc{PSpace}=\textsc{NPSpace} and show that any formula has a saturated tableau that terminates with a longest branch of polynomial length. We specify that the length of the root sentence for the tableau $\alpha$, denoted  $|\alpha|=m$, is determined by accumulating the number of symbols occurring in the sentence, and the number of symbols occurring in the standpoint modality indexes. For example, the sentence $\Box_s p$ has length $3$, while $\Box_{s\cup t} p$ has length $5$.\\

    We then note the following: 
    \begin{itemize}
        \item The only rules in the tableau calculus which create a new label (possible world) in the denominator are $(\Diamond_e)$, $(\dposs_e)$, $(\leadsto)$, $(\lesssim)$, $(\lesssim^+)$, $(\not\lesssim)$, $(\dnec)$, $(\dnec^+)$, and \textbf{(non-empty $\SSS$)}. In the case of all rules except $(\lesssim)$, and \textbf{(non-empty $\SSS$)}, the labeled sentence which occurs in the denominator is shorter than the one before, and so if the length of the root formula is $m$, these rules can only be applied $m$ times in total. In the case of $(\not\lesssim)$, we note that when $n::\neg (e\lesssim d)$ occurs in a branch, the rule can only be applied once for any pair of standpoint expressions $e,d\in\mathcal{E}$. This is because, once the denominator of the rule is applied, the rule holds relative to any label.
        \item Similarly, the number of atomic standpoints which occur in the formula cannot be greater than $m$, and so $\textbf{(non-empty $\mathcal{S}$)}$ can only be applied $m$ times. Therefore the total number of labels generated by the last two points are in $O(m^2)$.
        \item In the case of $(\lesssim)$ and $(\lesssim^+)$, each rule can only be applied once to the labels in the branch which are generated by other rules. Importantly, the new label $n^\star$ which both generate already fulfills a semantic conditions in the denominator of $(\lesssim)$ and  $(\lesssim^+)$ (since $n^\star\in W_d$) and so by the conditions of applicability which we previously defined, when the presence of $n^\star\in W_{s_1},...,n^\star\in W_{s_k}$ and $n::e\lesssim d$ occurs, the rule $(\lesssim)$ or $(\lesssim^+)$ is not applicable. Hence, the number of new labels $(\lesssim)$ and $(\lesssim^+)$ generates is at most $4k$, where $k$ is the number of labels generated by other rules. Therefore, the number of labels is in $O(4m^2)=O(m^2)$.
        \item When a given world is allocated to $W_e$ for some complex standpoint expression $e$, then the length of this standpoint expression is bounded by $m$ and the so the rules $(\cup)$ and $\cap$ can only be applied at most $m$ times, while $(*_1)$, $(*_2)$, $(\bot_{-})$ and $(\bot_\prec)$ can only be applied once per world. Hence, the maximum number of applications of these rules are in $O(m)$.
        \item Then, for a given label, the number of new sentences allocated to this label by a modal quantifier (i.e. the number of applications of $\Box_c$, $\Box_{c^+}$, $\dnec$, $\dnec^+$) is also bound by $m$, since each application of such a rule shortens the labeled sentence, and so is in $O(m)$.
        \item Then, the last set of rules $(\wedge)$, $(\neg)$, $(\vee)$ $\Box_{e\cup d}$ and $(\bot)$ are those which produce labeled sentences in the denominator with the same label that occurs in the numerator. Since the labeled sentence in the denominator is always shorter than the labelled sentence in the numerator, these too can only be applied a maximum $m$ times per label, with the exception of $(\bot)$ whch can only be applied once per label. Hence this step is in $O(m)$ labels.
        \item Lastly, the length of each label is $O(m)$ in length.
    \end{itemize}
    
        Therefore, the depth of any branch of the tableau is in $O(m^6)$ and so the total non-deterministic space complexity is in $O(m^7)$. Thus, the algorithm runs in \textsc{PSpace}. 
\end{proof}

\begin{corollary}
The tableau calculus for $\LangdSt$ terminates.
\end{corollary}

\begin{proof}
    As per the above, any branch in the tableau calculus has finite depth. Furthermore, it is worth noting that any rule which applies multiple branches can only be applies finitely many times (as is shown in the points above), and at most such a rule produces $3$ branches, meaning that the total branches in the tableau are finite, and thus the tableau terminates.
\end{proof}

\noindent\textbf{Completeness Proof:}\\

\noindent In order to show completeness we show that if there exists an open (saturated) tableau for $\alpha$, then $\alpha$ is locally satisfiable. We then show from any open saturated tableau $\mathcal{T}$ for $\alpha\in\LangdSt$ that there is some SPSS such that $\llbracket\alpha\rrbracket\neq \emptyset$. We define this SPSS as follows.\\

Let $\mathcal{T}$ be an open (saturated) tableau. Then by Definition \ref{def:closed-and-open-tableau} there is some open branch $(\mathcal{B},\Sigma,\prec)$ in $\mathcal{T}$. We define the SPSS $M_\T:=(\Pi_\T,\sigma_\T,\gamma_\T,\prec_\T)$ by:
\begin{itemize}
    \item $\Pi_\T:=\{n\mid n::\beta\in\SSS\text{ or }n\in\Sigma(e)\text{ for some }e\in\mathcal{E}\}$
    \item $\sigma_\T(s):=\{n\mid n\in \Sigma(s)\}$ for all $s\in \mathcal{S}$, $\sigma_\T(*)=\Pi_\T$ and the value of $\sigma_\T(e)$ is determined inductively on all other $e\in\mathcal{E}$ by $\sigma_\T(e_1\cap e_2)=\sigma_\T(e_1)\cap \sigma_\T(e_2)$, $\sigma_\T(e_1\cup e_2)=\sigma_\T(e_1)\cup \sigma_\T(e_2)$  and  $\sigma_\T(-e_1)=\Pi_\T\setminus\sigma_\T(e_1)$.
    \item $\gamma_\T:\Pi_\T\rightarrow2^\mathcal{P}$ is a function where $\gamma_\T(n):=\{p\in \mathcal{P}\mid n::p\in\SSS\}$ for all $n\in \Pi_\T$.
    \item $\prec_\T=\prec^+$, the transitive closure of $\prec$ appearing in the branch $(\mathcal{B},\Sigma,\prec)$.
\end{itemize}

The following four Lemmas show that $M_\T$ is in fact an SPSS such that $\llbracket\alpha\rrbracket\neq \emptyset$. In the proofs, we may ommit the suffix $\T$ from $M_\T:=(\Pi_\T,\sigma_\T,\gamma_\T,\prec_\T)$ where it is clear.

\begin{lemma}\label{lemma:saturated-branch-is-SPSS}
    $M_\T:=(\Pi_\T,\sigma_\T,\gamma_\T,\prec_\T)$ is a state-preferential standpoint structure.
\end{lemma}

\begin{proof}
    Clearly $\Pi_\T$ and $\gamma_\T$ are defined appropriately for an SPSS. Furthermore, since $\sigma_\T$ is defined inductively on complex standpoints and since $\sigma_\mathcal{T}(*)=\Pi_\T$, it respects all the identities in Definition \ref{def:standpoint-structure}. Moreover, since $\T$ is saturated, for every $s\in\SSS$ the application of rule \textbf{(non-empty $\SSS$)} requires that there exists some $n\in \Sigma(s)$. Hence $\sigma(s)\neq \emptyset$ for all $s\in\SSS$, and $\sigma_\T$ is well defined.

    We then only need to show that $\prec_\T$ is a strict partial order which satisfies well-foundedness. $\prec_\T$ is irreflexive, since no rule in the tableau calculus will assign $(n,n)$ to $\prec$ for any label $n$. Since $\prec_\T$ is a transitive closure it is transitive by definition. Well-foundedness follows from the fact that the tableau does not introduce any infinite descending chains: the only rules which which introduce pairs $(n,n')$ to $\prec$ are $(\dposs_e)$, $(\leadsto)$,, $(\not\leadsto)$, $\dnec$, $\dnec_c$, $\dnec_{c^+}$, $(\lesssim)$, $(\lesssim^+)$, and $(\not\lesssim)$. 
    Since these rules can only be applied finitely many times, all descending chains are finite and so, for any subset of $\Pi_\T$, $\prec$ must have minimal elements.
\end{proof}

\begin{lemma}\label{lemma:completenesss-standpoints-in-right-sigma-place}
    If $n\in \Sigma(e)$ appears in the branch $(\mathcal{B},\Sigma,\prec)$, then $n\in \sigma_\T(e)$. Moreover, if $e\in \SSS$ then the converse holds.
\end{lemma}

\begin{proof}
    We do this inductively, first considering the case where $e$ is a literal standpoint: if $e\in \SSS$ then by definition $\sigma(e)=\{n\mid n\in \Sigma(e)\}$ and so in this case  $n\in \sigma(e)$ iff $n\in\Sigma(e)$. Then if $e=-s$ for $s\in \SSS$ then if $n\in \Sigma(-s)$ we cannot have $n\in \Sigma(s)$ since the branch $\B$ is open and saturated (and having $n\in \Sigma(s)$ would require the addition of $n::\bot$ by rule $\bot_-$). Therefore $n\notin\sigma(s)$, and so $n\in\sigma(-s)$. We also consider the special case where $e=*$. Then, supposing that if $n\in \Sigma(e_1)$ implies $n\in \sigma(e_1)$ and $n\in \Sigma(e_2)$ implies $n\in \sigma(e_2)$ we have the following by induction:
    \begin{itemize}
          \item If $n\in \Sigma(e_1\cap e_2)$, then by rule $(\cap)$ we have $n\in \sigma(e_1)$ and  $n\in \sigma(e_2)$. Thus $n\in \sigma(e_1)\cap \sigma(e_1)=\sigma(e_1\cap e_2)$.
        \item If $n\in \Sigma(e_1\cup e_2)$, then by rule $(\cup)$ either $n\in \sigma(e_1)$ or  $n\in \sigma(e_2)$. In either case,  $n\in \sigma(e_1)\cup \sigma(e_1)=\sigma(e_1\cup e_2)$.
    \end{itemize}
Then, since each standpoint symbol is in disjunctive normal form, this is sufficient for all standpoints which will be introduced in $\Sigma$ during the tableau process.
\end{proof}

\begin{lemma}\label{lemma:PI-is-Sigma*}
    For any open saturated tableau, we have $\Pi_\T=\Sigma(*)$.
\end{lemma}

\begin{proof}
    If $n\in \Pi_\T$, by $(*_1)$ and $(*_2)$ we have that $n\in \Sigma(*)$. By definition if $n\in \Sigma(*)$, then  $n\in\Pi_\T$
\end{proof}

\begin{lemma}\label{lemma:big-completeness}
    If $\beta$ is a subsentence of $\alpha$, then whenever $n::\beta \in\SSS$, $n'\in \llbracket\beta\rrbracket^{M_\T}$.
\end{lemma}

\begin{proof}
    We do this using structural induction on the number of connectives in $\beta$. We shorten  $\llbracket \delta\rrbracket^{M_\T}$ to $\llbracket \delta\rrbracket$ for all $\xi\in\LangdSt$.

    In the base case $\beta$ is a literal. We then have two cases: (i) If $\beta=p\in\mathcal{P}$ then $n::p\in \B$ if and only if $p\in\gamma(n)$, and therefore $\llbracket p\rrbracket=\llbracket \beta\rrbracket$. (ii) Suppose $\beta=\neg p$ for some $p\in\mathcal{P}$, then $n::\neg p\in \B$. Then since $\T$ is saturated if we had $n::p\in\B$ we would have, by $(\bot)$ that $\B$ is closed. It then must be the case that $n::p\notin\B$, and therefore $p\notin\gamma(n)$. Hence $n\in\llbracket \neg p\rrbracket=\llbracket \beta\rrbracket$.

   Inductive step: For the inductive hypothesis (IH) we assume that if $n'::\delta\in\B$ then $n'\in \llbracket\delta\rrbracket$. The Boolean cases are done as usual, and the proofs for subsentences of the form $(\delta\leadsto \delta')$ or $\neg(\delta\leadsto \delta')$ are the same as the proof in Britz and Varzinczak \cite{britzvarzin:defeasiblemodalities}. We then consider the other cases for modal statements and defeasible standpoint sharpenings:

   \begin{itemize}
       \item $\beta=\neg\Box_e \delta$. If $n::\neg\Box_e \delta\in \B$, then by $(\Diamond_e)$ there exists $n'\in \Sigma(e)$ such that $n'::\neg \delta\in B$ by  (IH) we have $n'\in \llbracket\delta\rrbracket$ and by Lemma \ref{lemma:completenesss-standpoints-in-right-sigma-place}, we have $n'\in \sigma(e)$. Therefore $M,n\Vdash \neg\Box_e \delta$ and so $n\in \llbracket\delta \rrbracket$.
       \item $\beta=\neg\dnec_e \delta$. if $n::\neg\dnec_e \delta\in\B$ then by $\dposs_e$ there exists some $n'\in\min_\prec\Sigma(e)$ such that $n'::\neg \delta\in\B$. By Lemma \ref{lemma:completenesss-standpoints-in-right-sigma-place} and the definition of $\prec^\T$ we have that $n'\in\min_{\prec^\T}\sigma(e)$ and by IH we have $n'\in\llbracket \neg \delta\rrbracket$. Hence $M,n\Vdash \neg\dnec_e\delta$ and so $n\in \llbracket\neg\dnec_e\delta\rrbracket$.
       \item $\beta=e\lesssim d$. Let $e=c_1\cup...\cup c_m$ where all $c_i$ are conjuncts. If $n::c_1\cup...\cup c_m\lesssim d\in \B$ then for all $n'\in \min_\prec\sigma(c_1\cup...\cup c_2)$ we have $n'\in \sigma(c_i)$ for some $i$ with $1\leq i\leq m$. If $c_i=s_1\cap...\cap s_k\cap -s_{k+1}\cap...\cap s_{l}$ we have that $n\in \sigma(s_1)\cap...\cap \sigma(s_k)\cap \sigma(-s_{k+1})\cap...\cap \sigma(s_{l})\subseteq \sigma(s_1)\cap...\cap\sigma(s_k)=\Sigma(s_1)\cap...\cap\Sigma(s_k)$ (Lemma \ref{lemma:completenesss-standpoints-in-right-sigma-place}). Then by $(\lesssim)$ either (1) $n'\in \Sigma(s_{k+1}\cup...\cup s_l)$, (2) $n'\in \Sigma(d)$ or (3) there exists $n''$ such that $n''\prec n'$ and $n''\in \Sigma(c_1\cup...\cup c_m)\cap \Sigma(d)$. In case (1) by Lemma \ref{lemma:completenesss-standpoints-in-right-sigma-place}, $n'\in \sigma(s_{k+1}\cup...\cup s_l)$. Equivalently $n'\notin \sigma (-s_{k+1})\cap...\cap \sigma(-s_l)$ which contradicts the assumption that $n'\in\sigma(c_i)$. In case (3) $n'$ is not minimal in $\sigma(c_1\cup... \cup c_m)$, which contradicts the assumption that $n'\in \min_\prec\sigma(c_1\cup...\cup c_2)$. Therefore, case (2) must hold and so $n'\in \sigma(d)$ by Lemma \ref{lemma:completenesss-standpoints-in-right-sigma-place}. If $c_i=s_1\cap..\cap s_k$ (it has no negative standpoint literals), then we apply Rule $(\lesssim^+)$ and by a similar proof (without considering case (2)) we get $n'\in \sigma(d)$. Since $e$ is in DNF this covers all cases, and so $\min_\prec\sigma(e)\subseteq\sigma(d)$ and so $M,n\Vdash e\lesssim d$ and $n\in \llbracket e\lesssim d\rrbracket$.
       \item $\beta=\neg(e\lesssim d)$. If $n::\neg(e\lesssim d)\in \B$, then by $(\not\lesssim)$ there exists some $n'\in \Sigma(-d)$ set as minimal in $\Sigma(e)$. Then by Lemma \ref{lemma:completenesss-standpoints-in-right-sigma-place} we have $n'\in \pi\setminus \sigma(d)$ and since $n'$ is minimal in $\Sigma(e)$ then $n'\in \min_{\prec^\T}\sigma(e)$. Therefore $\min_{\prec^\T}\sigma(e)\nsubseteq \sigma(d)$, and so $M,n\nVdash e\lesssim d\implies M,n\Vdash \neg(e\lesssim d)\implies n\in \llbracket\neg(e\lesssim d)\rrbracket$.
       \item $\beta=\Box_e\delta$. We consider four cases:
       \begin{itemize}
           \item[i.] $e=*$. If $n::\Box_* \delta$, by Lemma \ref{lemma:PI-is-Sigma*} we have that $n'\in\Sigma(*)$ for all $n'\in \Pi$. Then by a special case of rule $(\Box_{c^+})$ we have that $n'::\delta\in \B$ for every $n'\in \Pi=\sigma(*)$ and by IH $n'\in \llbracket\delta \rrbracket$. Therefore, $M,n\Vdash \Box_* \delta$ so $n'\in \llbracket \Box_* \delta \rrbracket$.
           \item[ii.] $e=s_1\cap...\cap s_k$ for $s_i\in\SSS$. If $n::\Box_e\delta\in \B$ then by $(\Box_{c^+})$, for every $n'$ such that $n\in \Sigma(s_1),... n\in \Sigma(s_k)$ we have that $n'::\delta\in\B$ and $n'\in \llbracket\delta \rrbracket$ by IH. Moreover, for every $n''\in \sigma(s_1\cap...\cap s_k)=\sigma(s_1)\cap...\cap\sigma(s_k)$ we that by definition $n''\in \Sigma(s_1)\cap...\cap\Sigma(s_k)$. Hence, $n''\in \llbracket\delta \rrbracket$ for all $n''\in \sigma(s_1\cap...\cap s_k)$ and therefore $M,n\Vdash \Box_e\delta$ and $n\in \llbracket\Box_e\delta\rrbracket$.
           \item[iii.] $e=s_1\cap...\cap s_k\cap -s_{k+1}\cap...\cap -s_{m}$ for $s_i\in\SSS$. If $n::\Box_e \delta$, then, for any $n'\in \sigma(s_1\cap...\cap s_k\cap -s_{k+1}\cap...\cap -s_{m})=\sigma(s_1)\cap...\cap \sigma(s_k)\cap \sigma(-s_{k+1})\cap...\cap \sigma(-s_{m})$ we have by definition that $n'\in \Sigma(s_1)\cap....\Sigma(s_k)$. Then by rule $(\Box_c)$ either (1) $n'\in \Sigma(s_{k+1}\cup...\cup s_m)$ or (2) $n'::\delta \in \B$. Since $n'\in \sigma(-s_{k+1})\cap...\cap \sigma(-s_{m})$ then $n'\notin \sigma(s_i)$ for all $k+1\leq i\leq m$, and it follows from Lemma \ref{lemma:PI-is-Sigma*} that $n'\notin \Sigma(s_{i})$ for any such $i$. But then in case (1), by rule $(\cup)$ applied finitely many times, we must have that $n'\in \Sigma(s_{i})$ for some $k+1\leq i\leq m$, which is a contradiction. Therefore case (2) must hold, in which case $n'\in \llbracket\delta\rrbracket$ and therefore $M,n\Vdash \Box_e \delta\implies n\in \llbracket\Box_e \delta\rrbracket$. The above shows us that for all standpoint symbols $c$ which are conjuncts, then $n::\Box_c\delta\in\B$ implies $n\in\llbracket\Box_c\delta\rrbracket$.
           \item[iv.] $e=c_1\cup..\cup c_m$ where each $c_i$ is a conjunct. If $n::\Box_e\delta$, a finite number of applications of rule $(\Box_{e\cup d})$  gives us $n::\Box_{c_i}\delta\in \B$ whenever $1\leq i\leq m$. Then, by points i.-iii. we have that $n\in\llbracket\Box_{c_1}\delta\rrbracket\cap...\cap \llbracket\Box_{c_m}\delta\rrbracket$. Therefore, for any $n'\in \sigma(c_1\cup...\cup c_m)=\sigma(c_1)\cup...\cup \sigma(c_m)$ we have that $n'\in \sigma(c_i)$ for some $1\leq i\leq m$. Therefore, since $M,n\Vdash \Box_{c_i}\delta$, we have $M,n'\Vdash \delta$ and so $M,n\Vdash \Box_{c_1\cup..\cup c_m}\delta$. Therefore, $n\in \llbracket\Box_{c_1\cup..\cup c_m}\delta\rrbracket$.
       \end{itemize}
       Since each standpoint index in $\alpha$ is in DNF this suffices to cover all cases.
    \item $\beta=\dnec_e\delta$. We consider similar cases for $e$:
    \begin{itemize}
        \item[i.] $e=*$. If $n::{\dnec_*} \delta\in \B$, then for all $n'\in \Pi=\Sigma(*)$, by $(\dnec^+)$ either (1) $n'::\delta \in \B$ or (2) there exists a label $n''$ such that $n''\prec n'$, $n''\in \Sigma(*)$ and $n''::\delta\in\B$. Then for any $n'''\in \min_{\prec^\T}\sigma(*)$, case (b) cannot hold and so by (a), we have $n'''::\delta\in\B$. By IH then $n'''\in\llbracket\delta\rrbracket$ and so $M,n\Vdash{\dnec_*} \delta\implies n\in \llbracket{\dnec_*} \delta\rrbracket$.
        \item[ii.]  $e=s_1\cap...\cap s_k$ for $s_i\in\SSS$. If $n::\dnec_e\delta$ then for any $n'\in \min_{\prec^\T}\sigma(s_1\cap...\cap s_k)$ we have that $n'\in \sigma(s_1)\cap...\cap(s_1)$ and so $n'\in \Sigma(s_1)\cap...\cap\Sigma(s_k)$. Therefore, by $(\dnec^+)$ we have either (1) $n'::\delta\in \B$ or (2) there exists $n''$ such that $n''\prec n'$, $n''\in \Sigma(s_1)...n'\in\Sigma(s_k)$ and $n''::\delta\in\B$. In case (2) $n''\in \sigma(s_1)\cap...\cap\sigma(s_k)=\sigma(s_1\cap...\cap s_k)=\sigma(e)$ and then since $n'$ is assumed to be minimal in $\sigma(e)$ this is a contradiction. Therefore case (1) holds and $n'::\delta\in \B$ and by IH $n'\in \llbracket\delta\rrbracket$. Thus  $M,n\Vdash \dnec_e\delta\implies n\in \llbracket\dnec_e\delta\rrbracket$.
        \item[iii.] $e=s_1\cap...\cap s_k\cap -s_{k+1}\cap...\cap -s_{m}$ for $s_i\in\SSS$. Let $n::\dnec_e\delta\in\B$ and let $n'$ be a random member of $\min_{\prec^\T}\sigma(e)$. Then $n'\in \sigma(s_1)\cap...\cap \sigma(s_{k})\cap\sigma(-s_{k+1})\cap...\cap\sigma(-s_m)$. Then by Lemma \ref{lemma:completenesss-standpoints-in-right-sigma-place} $n'\in\Sigma(s_1)\cap...\cap \Sigma(s_{k})$ and so by $(\dnec)$ either (1) $n'\in \Sigma(s_{k+1}\cup...\cup s_m)$, (2) $n'::\delta\in \B$ or (3) there exists a label $n''$ such that $n''::\delta$, $n''\prec n'$, and $n''\in \Sigma(s_1)\cap...\cap\Sigma(s_k)\cap\Sigma(-s_{k+1})\cap...\cap\Sigma(-s_m)$. In case (1) by Lemma \ref{lemma:completenesss-standpoints-in-right-sigma-place} $n'\in\sigma(s_{k+1}\cup...\cup s_m)=\sigma(s_{k+1})\cup...\cup \sigma(s_m)$, but then $n'\in \sigma(s_i)$ where $k+1\leq i\leq m$. However, this contradicts the assumption that $n'\in \sigma(-s_{k+1})\cap...\cap\sigma(-s_m)$. In case (3) by Lemma \ref{lemma:completenesss-standpoints-in-right-sigma-place} $n''\in \sigma(s_1)\cap...\cap \sigma(s_{k})\cap\sigma(-s_{k+1})\cap...\cap\sigma(-s_m)=\sigma(s_1\cap...\cap s_{k}\cap-s_{k+1}\cap...\cap -s_m)=\sigma(e)$. However, this contradicts the assumption that $n'$ is minimal in $\sigma(e)$. Therefore, case (1) must hold and so $n'::\delta\in\B$, and by IH $n'\in\llbracket\delta\rrbracket$. Thus, $M,n\Vdash \dnec_e\delta\implies n\in \llbracket\dnec_e\delta\rrbracket$.
        \item[iv.] $e=c_1\cup..\cup c_m$ where each $c_i$ is a conjunct. Assume $n::\dnec_e\delta$, then for any $n'\in \min_{\prec^\T}\sigma(e)$ we have the following. $n\in \sigma(c_1\cup...\cup c_m)=\sigma(c_1)\cup...\cup \sigma(c_m)$, therefore $n'\in \sigma(c_i)$ for some conjunct $c_i$. Then either (a) $c_i$ is of the form $c_i=s_1\cap...\cap s_k\cap -s_{k+1}\cap...\cap -s_{m}$ or (b) $c_i$ is of the form $c_i=s_1\cap...\cap s_k$. In case (a) $n'\in \sigma(s_1)\cap...\cap\sigma(s_k)\cap\sigma(-s_{k+1})\cap...\cap\sigma(-s_l)=\Sigma(s_1)\cap...\cap\Sigma(s_k)\cap\sigma(-s_{k+1})\cap...\cap\sigma(-s_l)\subseteq \Sigma(s_1)\cap...\cap\Sigma(s_k)$ by Lemma \ref{lemma:completenesss-standpoints-in-right-sigma-place}. Then by $\dnec$ either (1) $n'::\delta\in \B$, (2) $n'\in \Sigma(s_{k+1}\cup...\cup s_l)$ or (3) there exists some label $n''$ such that $n''::\delta\in\B$, $n''\prec n'$ and $n''\in \Sigma(c_1\cup...\cup c_m)$. In case (2), we derive a contradiction using a similar proof to case (2) in point iii. In case (3) we have by Lemma \ref{lemma:completenesss-standpoints-in-right-sigma-place} that $n''\in \sigma(c_1\cup...\cup c_m)=\sigma(e)$ but then since $n'$ is assumed to be minimal in $\sigma(e)$ a contradiction results from the addition of $n''\prec n'$. Then case (1) must hold and $n'::\delta\in \B$. By IH $n'\in \llbracket\delta \rrbracket$ and so $M,n\Vdash \dnec_e\delta\implies n\in \llbracket\dnec_e\delta\rrbracket$. If case (b) holds then the proof that $n\in \llbracket\dnec_e\delta\rrbracket$ follows a similar method by applying rule $(\dnec^+)$, where we just do not need to consider branch (2) from the previous case.
    
        \end{itemize}
            Since every standpoint index is in DNF, i.-iv. cover all cases. 
   \end{itemize}
\end{proof}

\begin{lemma}[Completeness]\label{lemma:completenesss}
  For any $\alpha\in\LangdSt$, if there is an open saturated tableau for $\alpha$, then $\alpha$ is locally satisfiable.  
\end{lemma}

\begin{proof}
    If there exists an open saturated tableau $\T$ for $\alpha$, then $M^\T=(\Pi^\T,\sigma^\T,\gamma^\T,\prec^\T)$ is a well-defined SPSS by Lemma \ref{lemma:saturated-branch-is-SPSS} and $\llbracket\alpha\rrbracket\neq \emptyset$ by Lemma \ref{lemma:big-completeness}.
\end{proof}

\noindent\textbf{Soundness Proof:}\\

\noindent In order to show soundness, we must show that if $\alpha\in\LangdSt$ is satisfiable in the semantics of $\LangdSt$, then there exists an open saturated taleau for $\alpha$. Equivalently, if all the saturated tableau for $\alpha$ are closed, then $\alpha$ is unsatisfiable. 

\begin{definition}
    Let $\mathcal{B}$ be a set of labelled sentences. $\mathcal{B}(n):=\{\beta\mid n::\beta\in \mathcal{B}\}$, and $\widehat{\mathcal{B}(n)}=\bigwedge\mathcal{B}(n)$.
\end{definition}

\begin{lemma}\label{lemma:big-soundness-lemma}
    For every tableau rule in Figure \ref{fig:etableau-rules-for-satisfiability}, if for all $\mathcal{T}^{j+1}=\{...,(\mathcal{S}^{j+1},\Sigma^{j+1}, \prec^{j+1})...\}$ that can be obtained from $\mathcal{T}^{j}=\{...,(\mathcal{S}^{j},\Sigma^{j}, \prec^{j})...\}$ there is an $n$ such that $\widehat{\mathcal{S}^{j+1}(n)}$ is unsatisfiable in the semantics for $\LangdSt$, then there is an $n'$ such that $\widehat{\mathcal{S}^{j}(n')}$ is unsatisfiable.
\end{lemma}

\begin{proof}

We will not show proofs for rules $(\wedge)$, $(\neg)$, $(\vee)$ and $(\bot)$ since these have been considered in other cases. Rules $(\dposs_e)$, $(\leadsto)$, $(\bot_\prec)$ and $(\not\leadsto)$ are considered by Britz and Varzinczak \cite{britzvarzin:defeasiblemodalities}. The rule $\Diamond_e$ also refers to one specific case of the similar rule $(\Diamond_i)$ in Britz and Varzinczak's \cite{britzvarzin:defeasiblemodalities} work, and so we omit this too:
    \begin{itemize}
    \item \textbf{Rule $(\cap)$:} If $n\in \Sigma^j(e\cap d)$ then an application of $(\cap)$ extends $\Sigma^j$ by adding $n\in \Sigma^{j+1}(e)$ and $n\in \Sigma^{j+1}(d)$. Then, if $\widehat{\B^{j+1}(n')}$ is unsatisfiable, then in the case where $n'\neq n$ we have that $\B^j(n')=\B^{j+1}(n')$ and that all other semantic conditions concerning $n'$ are the same. Therefore $\widehat{\B^j(n')}$ is unsatisfiable. If $n=n'$ then the contradiction must come about from the addition of $n\in \Sigma^{j+1}(e)$ or $n\in \Sigma^{j+1}(d)$. Without losing generality, assume $n\in \Sigma^{j+1}(e)$ leads to unsatisfiability. Then, for any SPSS $M$ with $\pi_n$ such that $M,\pi_n\Vdash \widehat{\B^{j}(n)}$, we must have that $\pi_n\notin \sigma(e)$. However, clearly in any SPSS $\vDash e\cap d\leq e$. Therefore since $n\in \Sigma^j(e\cap d)$ but in any SPSS no precisification satisfying $\widehat{\B^{j}(n)}$ can be in $\sigma(e)$ implies that $\widehat{\B^{j}(n)}\vDash \neg (e\cap d\leq e)$ and so $\widehat{\B^{j}(n)}$ is unsatisfiable.
    \item \textbf{Rule $(\cup),(*_1),(*_2),(\bot_{-})$} have similar proofs to $(\cap)$ using the tautologies $e\leq e\cup d$ for $(\cup)$, $e\leq *$ for $(*_1)$ and $(*_2)$, and $\Box_{e\cap -e}\bot=e\leq e$ for $(\bot_{-})$.
    \item \textbf{Rule $(\Box_{e\cup d})$:} If $\B^j$ contains $n::\Box_{e\cup d}\beta$ then applying $(\Box_{e\cup d})$ adds $n::\Box_e\beta$ and $n::\Box_d\beta$ to $\B^{j}$. Then assume $\widehat{\B^{j+1}(n')}$ is unsatisfiable. As in the previous case, if $n\neq n'$, the Lemma follows. Then if $n=n'$ either $\widehat{\B^{j}(n)}\vDash \neg \Box_e\beta$ or $\widehat{\B^{j}(n)}\vDash \neg \Box_e\beta$ in order to derive that $\widehat{\B^{j+1}(n')}$ is unsatisfiable. Assume, the first case. Then in every SPSS $M$ such that $\widehat{\B^{j}(n)}$ holds at some precisificaton $\pi_n$ (I.e. $M,\pi_n\Vdash \widehat{\B^{j}(n)}$) there exists some $\pi'\in \sigma(e)\subseteq \sigma(e\cup d)$ such that $M,\pi_n\Vdash \neg \alpha$. But then for every such $M$ we have $M\Vdash \neg \Box_{e\cup d} \beta$, and so $\widehat{\B^{j}(n)}\vDash \neg\Box_{e\cup d}\beta$, since this occurs in any SPSS with a precisification satisfying satisfying $\widehat{\B^{j}(n)}$. Therefore, since $\Box_{e\cup d}\beta\in \B^{j}(n)$, we must have $\widehat{\B^{j}(n)}$ is unsatisfiable.
    \item \textbf{Rule $(\Box_c)$:} If $\B^j$ contains $n::\Box_{s_1\cap...\cap s_k\cap -s_{k+1}\cap...\cap -s_m}\beta$ and $n'\in \Sigma(s_1)\cap...\cap\Sigma(s_k)$ then an application of $(\Box_c)$ gives us either:
    \begin{itemize}
        \item $\T^{j+1}_{(1)}$ where the labelled sentences are the same but $n'\in \Sigma ^{j+1}_{(1)}(s_{k+1}\cup...\cup s_m)$ is added to $\Sigma^j$
        \item $\T^{j+1}_{(2)}$ where we add $n'::\beta$ to obtain $\B^{j+1}_{(2)}$. 
    \end{itemize}
    Assume there exists $n_1$ or $n_2$ such that $\widehat{\B^{j+1}_{(1)}(n_1)}$ and $\widehat{\B^{j+1}_{(2)}(n_2)}$ are both unsatisfiable. If either $n_1\neq n'$ or $n_2\neq n'$ then by similar reasoning to Rule $(\cap)$ the Lemma holds (even if $n_1$ or $n_2$ are equal to $n$, since we do not add extra labelled sentences or semantic conditions to $n$ the argument still applies). Therefore we consider the case where $n_1=n'=n_2$. If $\widehat{\B^{j+1}_{(1)}(n')}$ is unsatisfiable, since we add no labeled sentences to $\B^j$, we must have that the contradiction (unsatisfiability) derives from $n'\in \Sigma ^{j+1}_{(1)}(s_{k+1}\cup...\cup s_m)$. In this case, that means that for any SPSS $M$ containing a precisification $\pi_{n'}$ such that $M,\pi_{n'}\Vdash \widehat{\B^j(n')}$ and $\pi_{n'}\in \sigma(s_1)\cap...\cap\sigma(s_k)$ we must have that $\pi_{n'}\notin \sigma(s_{k+1}\cup...\cup s_m)$. Equivalently, $\pi_{n'}\in \sigma(-s_{k+1}\cap...\cap -s_m)$. Then since $\widehat{\B^{j+1}_{(1)}(n_2)}$ is unsatisfiable, we must have $\widehat{\B^j(n')}\vDash \neg \beta$.  But then, in any SPSS $M$ containing a precisification $\pi_{n'}$ such that $M,\pi_{n'}\Vdash \widehat{\B^j(n')}$ and $\pi_{n'}\in \sigma(s_1)\cap...\cap\sigma(s_k)$ we have that both $M,\pi_{n'}\Vdash \neg \beta$ and $\pi_{n'}\in \sigma(-s_{k+1}\cap...\cap -s_m)$. But then $\pi_{n'}\in \sigma(s_1\cap...\cap s_k\cap -s_{k+1}\cap...\cap s_m)$. But then for any SPSS, $\vDash \neg\Box_{s_1\cap...\cap s_k\cap -s_{k+1}\cap...\cap -s_m}\beta$ and so $\widehat{\B^j(n)}\vDash \neg \Box_{s_1\cap...\cap s_k\cap -s_{k+1}\cap...\cap -s_m}\beta$ and since $\Box_{s_1\cap...\cap s_k\cap -s_{k+1}\cap...\cap -s_m}\beta\in \B^j(n)$ we have that $\widehat{\B^j(n)}$ is unsatisfiable.
    \item \textbf{Rule $(\Box_{c^+})$} has a similar proof to $(\Box_c)$ except the branch $\T^{j+1}_{(2)}$ does not need to be considered.
    
    \item \textbf{Rule $(\dnec)$:} If $\B^j$ contains $n::\dnec_{c_1\cup...\cup c_m}\beta$ where each $c_i$ is a conjunct and $n'\in \Sigma(s_1)\cap...\cap\sigma(s_k)$ where $c_i=s_1\cap...\cap s_k\cap s_{k+1}\cap...\cap s_l$, then applying rule $(\dnec)$ gives us:
    \begin{itemize}
        \item $\T^{j+1}_{(1)}$ where $n'::\beta$ is added to obtain $\B^{j+1}_{(1)}$.
        \item  $\T^{j+1}_{(2)}$ where $\Sigma^{j}$ is extended to have $n'\in \Sigma^{j+1}_{(2)}(s_{k+1}\cup...\cup s_l)$.
        \item $\T^{j+1}_{(3)}$ where a fresh label $n''$ is added, $n''::\beta$ is added to $\B^{j+1}_{(3)}$ and $\Sigma^j$ and $\prec^j$ are extended by  $n''\in\Sigma^{j+1}_{(3)}(c_1\cup...\cup c_m)$ and $n''\prec^{j+1}_{(3)}n'$.
    \end{itemize}
   Suppose then there exist $n_1$, $n_2$ and $n_2$ such that $\widehat{\B^{j+1}_{(1)}(n_1)}$, $\widehat{\B^{j+1}_{(2)}(n_2)}$ and $\widehat{\B^{j+1}_{(3)}(n_3)}$ are unsatisfiable. As with previous cases, if $n_1\neq n'$, $n_2\neq n'$ or $n_3\notin\{n',n''\}$ the Lemma holds. We consider the case where $n_1\neq n'$, $n_2\neq n'$ and $n_3\in\{n',n''\}$. For similar reasons to Rule $(\Box_c)$ $\widehat{\B^{j+1}_{(2)}(n')}$ being unsatisfiable means that for any SPSS $M$ with a precisification $\pi_{n'}$ such that $M,\pi_{n'}\Vdash \widehat{\B^j(n')}$ and $\pi\in \sigma(s_1\cap...\cap s_k)$  we have that $\pi_{n'}\in \sigma(-s_{k+1}\cap...\cap-s_l)$. Similarly, the unsatisfiablility of  $\widehat{\B^{j+1}_{(2)}(n')}$ leads to any SPSS $M$ with a precisification $\pi_{n'}$ such that $M,\pi_{n'}\Vdash \widehat{\B^j(n')}$ and $\pi\in \sigma(s_1\cap...\cap s_k)$ implies that $M,\pi_{n'}\Vdash \neg \beta$. We then consider the case where $n_3=n'$. Then the only condition added which affects $n'$ is that $n''\prec^{j+1}_{(3)}n'$. That is, the non-minimality of $n'$ leads to a semantic contradiction. Hence, in this case, for any any SPSS $M$ with a precisification $\pi_{n'}$ such that $M,\pi_{n'}\Vdash \widehat{\B^j(n')}$, we must have that there is no $\pi_{n''}\in \Pi$ such that $\pi_{n''}\prec \pi_{n'}$. Combining these three conditions, we have that for any SPSS $M$ such with a precisification $\pi_{n'}$ such that $M,\pi_{n'}\Vdash \widehat{\B^j(n')}$ and $\pi\in \sigma(s_1\cap...\cap s_k)$ then $\pi_{n'}\in\sigma(s_1\cap...\cap s_k)\cap -s_{k+1}\cap...\cap-s_l)=\sigma(c_i)\subseteq \sigma(c_1\cup...\cup c_m)$, $M\pi_{n'}\Vdash \neg \beta$ and $\pi_{n'}$ is minimal. That is $M\Vdash \neg \dnec_{c_1\cup...\cup c_m}\beta$, and in particular, for any SPSS $M$ with such a $\pi_{n'}$, we have that $M,\pi_{n}\Vdash \widehat{\B^j(n)}$ implies $M,\pi_{n}\Vdash \neg \dnec_{c_1\cup...\cup c_m}\beta$ (That is, under the semantic conditions of $\T^j$, $ \widehat{\B^j(n)}\vDash \neg \dnec_{c_1\cup...\cup c_m}\beta$). Then since $\dnec_{c_1\cup...\cup c_m}\beta\in\B^j(n)$ we have that $\widehat{\B^j(n)}$ is unsatisfiable.
    In the case where $n_3=n''$ then $\widehat{\B^{j+1}_{(3)}(n'')}$ is unsatisfiable, and so since $n''$ is a fresh label $\B^{j+1}_{(3)}(n'')=\{\beta\}$. Then, for any SPSS $M$, if a precisification $\pi_{n''}$ is in $\sigma(c_1\cup...\cup c_m)$ and is minimal (since $n''$ is fresh there is no element minmal to it), we must have $M,\pi_{n''}\Vdash \neg \beta$. That is, $\vDash \dnec_{c_1\cup...\cup c_m}\neg \beta$, which implies $\vDash \neg\dnec_{c_1\cup...\cup c_m} \beta$ unless $\sigma(c_1\cup...\cup c_m)$ is empty (which by this application of Rule $(\dnec_{cup})$ that introduces $n''$, cannot be true). Then $\widehat{\B^j(n)}$ is unsatisfiable since $\dnec_{c_1\cup...\cup c_m}\beta\in\B^j(n)$. 
    \item \textbf{Rule $(\dnec^+)$} has a similar proof to $(\dnec)$, but without having to consider the first branch, since $c_i$ only contains positive literals.

    \item \textbf{Rule $(\lesssim$):} Suppose $B^j$ contains $n::e\lesssim d$ where $e=c_1\cup ...\cup c_m$, such that each $c_i$ is a conjunct\footnote{Since $e$ is in DNF we can make this assumption.}. Then if $n'\in \Sigma(s_1)\cap...\Sigma(s_k)$ where some $c_i=s_1\cap ...\cap s_k\cap -s_{k+1}\cap...\cap s_l$, an application of rule $(\lesssim)$ gives us either:
    \begin{itemize}
        \item $\T^{j+1}_{(1)}$ where we extend $\Sigma^j$ by adding $n'\in \Sigma^{j+1}_{(1)}(s_{k+1}\cup...\cup s_l)$.
         \item $\T^{j+1}_{(2)}$ where we extend $\Sigma^j$ by adding $n'\in \Sigma^{j+1}_{(2)}(d)$.
        \item $\T^{j+1}_{(3)}$ where we add a fresh label $n''$ and extend $\Sigma^j$ and $\prec^{j}$ by adding $n''\prec^{j}_{(3)} n'$, $n'\in \Sigma^{j+1}_{(3)}(e)$ and $n'\in \Sigma^{j+1}_{(3)}(d)$.
    \end{itemize}
    Now assume there exists $n_1$, $n_2$ and $n_3$ such that $\widehat{\B^{j+1}_(1)(n_1)}$, $\widehat{\B^{j+1}_(2)(n_2)}$ and $\widehat{\B^{j+1}_{(3)}(n_3)}$ are all unsatisfiable. As in previous cases, if $n_1\neq n'$, $n_2\neq n'$ or $n_3\notin \{n',n''\}$ then the Lemma follows. We then consider the case where $n_1= n'$, $n_2= n'$ and either (A) $n_3=n'$ or (b) $n_3=n''$. In branch $(1)$ since we add no labelled sentences to $n'$, if $\widehat{\B^{j+1}_(1)(n')}$ is unsatisfiable, following a similar reasons to branch (2) for Rule $\dnec$, we have that for any SPSS $M$, if there exists $\pi_{n'}$ such that $M,\pi_{n'}\Vdash \widehat{\B^j(n')}$ and $\pi_{n'}\in \sigma(s_1\cap...\cap s_k)$, then $\pi_{n'}\in \sigma(-s_{k+1}\cap...\cap -s_l)$. Similarly, the unsatisfiability of $\widehat{\B^{j+1}_(2)(n')}$ implies that in any SPSS $M$, if there exists $\pi_{n'}$ such that $M,\pi_{n'}\Vdash \widehat{\B^j(n')}$ and $\pi_{n'}\in \sigma(s_1\cap...\cap s_k)$ then $\pi_{n'}\notin \sigma(d)$. Then in the case (a) where $\widehat{\B^{j+1}_{(3)}(n')}$ is unsatisfiable, since the only semantic condition added to $\T^j$ involving $n'$ is the fact that there exists $n''$ such that $n''\prec n'$, then it is the non-minimality of $n'$ that leads to contradiction. That is, for any SPSS $M$, if there exists $\pi_{n'}$ such that $M,\pi_{n'}\Vdash \widehat{\B^j(n')}$ and $\pi_{n'}\in \sigma(s_1\cap...\cap s_k)$ then $\pi_{n'}$ must be minimal in $\Pi$. However, if we put these three conditions together, then for any SPSS $M$, if there exists $\pi_{n'}$ such that $M,\pi_{n'}\Vdash \widehat{\B^j(n')}$ and $\pi_{n'}\in \sigma(s_1\cap...\cap s_k)$ then by (1), $\pi_{n'}\in \sigma(s_1\cap...\cap s_k\cap-s_{k+1}\cap...\cap -s_l)=\sigma(c_i)\subseteq \sigma(e)$, $\pi_{n'}\notin \sigma(d)$ and $\pi_{n'}$ is globally minimal and thus in $\min_\prec\sigma(e)$. Therefore, $\min_\prec\sigma(e)\nsubseteq\sigma(d)$ and so $M\Vdash \neg(e\lesssim d)$, and in particular $M,\pi_n\Vdash \neg(e\lesssim d)$ for any $\pi_n$ such that $M,\pi_n\Vdash \widehat{\B^j(n)}$. Then, under the semantic conditions of $\T^j$\footnote{Which we recall are constructed by the initial formula appearing at the root of the tableau.} we have $\widehat{\B^j(n)}\vDash \neg(e\lesssim d)$ and so $\widehat{\B^j(n)}$ is unsatisfiable. In case (b) $\widehat{\B^{j+1}_{(3)}(n'')}$ is unsatisfiable. Then since $n''$ is a fresh label $\B^{j+1}_{(3)}(n'')=\emptyset$. Therefore if $\widehat{\B^{j+1}_{(3)}(n'')}$ is unsatisfiable, we must that for any SPSS $M$ satisfying the conditions of $\T^j$ there can be no precisification $\pi_{n''}$ such that $\pi_{n''}\in \sigma(e)$, $\pi_{n''}\in \sigma(d)$ and (iii) $\pi_{n''}\prec \pi_{n'}$ for \textit{any} $\pi_{n'}\in \sigma(s_1\cap...\cap \sigma(s_k)$. Therefore, either (i) $\sigma(e)$ is empty, (ii) $\sigma(d)$ is empty, or (iii) every $\pi_{n'}\in \sigma(s_1\cap...\cap \sigma(s_k)$ is globally minimal. In case (i) we get a contradiction since $\T^j$ and the unsatisfiability of branch (1) means that there exists $\pi_{n'}\in \sigma(s_1\cap...\cap s_k\cap-s_{k+1}\cap...\cap -s_l)=\sigma(c_i)\subseteq \sigma(e)$. In case (ii), since $\sigma(e)$ is non-empty by the previous reasoning, and $\sigma(d)$ is empty. Therefore $\min_\prec \sigma(e)\nsubseteq \emptyset=\sigma(d)$ and we must have $M\Vdash \neg(e\lesssim d)$. Hence, by similar reasoning as in case (a), we have $\widehat{\B^j(n)}\vDash \neg(e\lesssim d)$ under the semantic conditions of $\T^j$, and so $\widehat{\B^j(n)}$ is unsatisfiable. In case (iii), we have that any $\pi_{n'}\in \sigma(s_1\cap...\cap \sigma(s_k)$ is globally minimal, then by branch (1) and (2) we have that $\pi_{n'}\in \sigma(s_1\cap...\cap \sigma(s_k)$ implies that $\pi_{n'}\in \sigma(e)$ and $\pi_{n'}\notin \sigma(d)$. Moreover, by the conditions of $\T^j$, $\sigma(s_1\cap...\cap \sigma(s_k)$ is non-empty. Then since any such $\pi_{n'}$ is globally minimal, it must be in $\min_\prec \sigma(e)$. Therefore $\min_\prec \sigma(e)\nsubseteq \sigma(d)$ and so following similar reasoning to case (ii) $\widehat{\B^j(n)}$ is unsatisfiable.
    \item\textbf{Rule $(\lesssim^+)$} is similar to Rule $(\lesssim)$ but without needing to consider branch $\T^{j+1}_{(1)}$.
        \item \textbf{Rule $(\not\lesssim)$}: If $n::\neg (e\lesssim d)\in \B^j$, then applying $(\not\lesssim)$ introdcues a fresh label $n'$ such that $n'$ is minimal in $\Sigma^{j+1}(e)$ and $n\in\Sigma^{j+1}(-d)$. Assume there exists $n''$ such that $\widehat{\B^j(n'')}$ is unsatisfiable. By the same reasoning, if $n''\neq n'$, the Lemma follows. Consider the case where $\widehat{\B^j(n')}$ is unsatisfiable. Since $n'$ is fresh  then $\B^j(n')=\emptyset$. Therefore, if $\widehat{\B^j(n')}$ is unsatisfiable, for any SPSS $M$ satisfying $\T^j$, we must have that there exissts no precisification $\pi_{n'}$ such that $\pi_{n'}$ is minimal in $\sigma(e)$ and $\pi_{n'}\notin \sigma(d)$. Therefore in such a case $\min_\prec \sigma(e)\subseteq \sigma(d)$ and $M\Vdash e\lesssim d$ and in particular $M,\pi_n\Vdash e\lesssim d$ for any precisification $\pi_n$ such that $M,\pi_n\Vdash \widehat{\B^j(n')}$. Therefore, if $\T^j$ holds, then $\widehat{\B^j(n')}\vDash e\lesssim d$, but since $n::\neg (e\lesssim d)\in \B^j$, then $\widehat{\B^j(n')}$ is unsatisfiable.
        \item \textbf{Rule (non-empty $\SSS$):} If $\T^j$ is such that $\Sigma^j(s)=\emptyset$ for some standpoint $s\in\SSS$ in our vocabulary, and no other rules are applicable, then we apply \textbf{(non-empty $\SSS$)} to obtain  $\T^{j+1}$ where a fresh label $n$ so that $\Sigma^{j+1}(s)=\{n\}$. If we assume that tere exists a label $n'$ such that $\widehat{\B^{j+1}(n')}$ is unsatisfiable, then if $n'\neq n$ for similar reasons as above the Lemma holds. On the other hand, if $n=n'$ then since $n$ is a fresh label and $\B^{j+1}(n')=\emptyset$, we have that any SPSS $M$ satisfying the conditions in $\T^j$ must have $\sigma(s)=\emptyset$. But then $M$ is not a well-defined SPSS, since by definition we require that $\sigma(s')\neq \emptyset$ and so $\T^j$ is unsatisfiable. In particular, for any label $n''$ in $\T^j$, by the above we have $\widehat{\B^j(n'')}\vDash *\leq -s$. But by definition every SPSS $M$ has $M\Vdash \neg(*\leq -s)$ and so $\widehat{\B^j(n'')}$ for any label $n''$ in $\T^j$.\footnote{We also know the set of labels in $\T^j$ are non-empty since the tableau for any sentence starts with the inclusion of label $0$.}
    \end{itemize} 
\end{proof}

\begin{lemma}[Soundness.]\label{lemma:soundness}
    If $\alpha\in\LangdSt$ is locally satisfiable, then there exists an open saturated tableau for $\alpha$.
\end{lemma}

\begin{proof}
    We prove the converse. By Lemma \ref{lemma:big-soundness-lemma} if all tableau for $\alpha$ are closed, then for any branch $(\B,\Sigma,\prec)$ there is some $n$ such that $\widehat{\B(n)}$ is unsatisfiable. Then in particular, $\widehat{\B(0)}=\alpha$ is unsatisfiable, since $0$ is the only label in the root of the tableau. That is, all rules in Figure \ref{fig:etableau-rules-for-satisfiability} preserve satisfiability.
\end{proof}

Together, these give us the results needed for our main result.\\

\noindent\textbf{Theorem \ref{theorem:soundness and completeness}. } \textit{The tableau algorithm is sound and complete with respect to local satisfiability in SPSS semantics.}

\begin{proof}
    Completeness is shown in Lemma \ref{lemma:completenesss}. Soundness is shown in Lemma \ref{lemma:soundness}.
\end{proof}
\end{document}